\documentclass[12pt]{iopart}
\pdfoutput=1
\usepackage[utf8]{inputenc}
\usepackage[english]{babel}
\usepackage[T1]{fontenc}
\usepackage{amsmath, bbold, enumerate,enumitem,xcolor, amssymb,amsthm}
\usepackage{hyperref}
\usepackage[sort&compress,numbers,square,comma]{natbib}
\usepackage{tikz}
\usepackage{lipsum}
\usepackage{wrapfig}
\usepackage{tablefootnote}
\usepackage{booktabs}
\usepackage{capt-of}
\newcommand{\clS}{\mathcal{S}}

\usepackage[normalem]{ulem}

\newcommand{\clSeto}{\mathcal{S}_{\rm ETO}}
\newcommand{\clSto}{\mathcal{S}_{\rm TO}}
\newcommand{\clSceto}{\mathcal{S}_{\rm CETO}}
\newcommand{\clScetoqb}{\mathcal{S}_{\rm CETO}^{(2)}}
\newcommand{\clScetoqd}[1]{\mathcal{S}_{\rm CETO}^{(#1)}}

\usepackage{geometry}
\usepackage{float}
\usepackage[most]{tcolorbox}
\usepackage{tikz}
\usepackage{varwidth}
\usepackage{wrapfig}

\newcommand{\border}{{\pi}}
\newcommand{\pstate}{\mathrm{p}}
\newcommand{\qstate}{\mathrm{q}}
\newcommand{\rstate}{\mathrm{r}}

\newcommand{\estate}{\mathrm{e}}
\newcommand{\slope}{\mathrm g}
\newcommand{\Mswap}{\mathcal{M}}
\newcommand{\catstate}{{\rm c}}
\DeclareMathOperator{\sgn}{sgn}
\newcommand{\probspace}{\mathcal{V}}

\newtheorem{theorem}{Theorem}

\newtheorem{remark}{Remark}

\newtheorem{lemma}[theorem]{Lemma}

\newtheorem{corollary}[theorem]{Corollary}

\newcommand{\ket}[1]{|#1\rangle}

\newcommand{\ketbra}[2]{| #1 \rangle \langle #2 |}

\newtheorem{definition}{Definition}[section]

\newcommand{\Eps}{\mathcal{E}}
\newcommand{\tbeta}{{\tau}^\beta}

\usepackage{tcolorbox}
\tcbuselibrary{breakable}
\tcbset{enhanced,colback=red!5!white,colframe=red!65!black,fonttitle=\bfseries, boxrule=0.7pt}
\newtcolorbox{mybox}[1]{colback=red!5!white,colframe=red!65!black,fonttitle=\bfseries,title=#1,boxrule=0.7pt}
\usepackage{makecell}
\newcolumntype{L}[1]{>{\raggedright\let\newline\\\arraybackslash\hspace{0pt}}p{#1}}

\setenumerate{topsep=5pt,itemsep=0ex,partopsep=1ex,parsep=1ex}

\begin{document}
\submitto{\NJP}

\title{Catalysis in Action via Elementary Thermal Operations}
\author{Jeongrak Son}
\address{School of Physical and Mathematical Sciences, Nanyang Technological University, 637371, Singapore}
\ead{jeongrak.son@e.ntu.edu.sg}
\author{Nelly H. Y. Ng}
\address{School of Physical and Mathematical Sciences, Nanyang Technological University, 637371, Singapore}
\ead{nelly.ng@ntu.edu.sg}

\begin{abstract}
We investigate catalysis in the framework of elementary thermal operations, leveraging the distinct features of such operations to illuminate catalytic dynamics. As groundwork, we establish new technical tools that enhance the computability of state transition rules for elementary thermal operations. Specifically, we provide a complete characterisation of state transitions for a qutrit system and special classes of initial states of arbitrary dimension. By employing these tools in conjunction with numerical methods, we find that by adopting a small catalyst, including just a qubit catalyst, one can significantly enlarge the set of state transitions for a qutrit system. This advancement notably narrows the gap of reachable states between elementary thermal operations and generic thermal operations. Furthermore, we decompose catalytic transitions into time-resolved evolution, which critically enables the tracking of nonequilibrium free energy exchanges between the system and bath. Our results provide evidence for the existence of simple and practicable catalytic advantage in thermodynamics while offering insight into analysing the mechanism of catalytic processes.
\end{abstract}
\maketitle

\section{Introduction}\label{sec:intro}

Catalysts are auxiliary states that enter a process in a way such that they are recovered at the end. They are particularly useful when one is restricted to a limited set of operations (motivated by fundamental principles or practical considerations). 
The general understanding is that the participation of such states helps in mitigating dynamical constraints, by expanding the working Hilbert space, without consuming resources in the catalyst. Since catalysts ideally suffer no deterioration, they can be used repeatedly to activate state transitions. 

Since the discovery of quantum catalysis, first reported in entanglement theory~\cite{Jonathan99_PRL, Eisert00_Cat, Daftuar01_trumping, Anspach01_condition, Klimesh07_ineqs, Kondra21_EntCat}, it has been extended to various quantum resource theories~\cite{Datta22_CatReview, LipkaBartosik23_catreview}, such as coherence~\cite{Aberg14_PRL, Bu16_coherence, Ding21_corr, Ryuji22_catalyst}, thermodynamics~\cite{Ng15_NJP, Brandao15_2ndlaws, Muller18_corr, Boes20_Fluc, Henao21_catalytic, Shiraishi21_PRL, Yadin22_Gaussian}, randomness~\cite{Boes18_PRX_rand, Boes19_catconjecture, Lie21_Cat1, Lie21_PRR, Wilming21_catconjecture}, and quantum teleportation~\cite{Lipka-Bartosik21_tele}. 
Previous studies on catalysis focus predominantly on the initial and final states of state transformation, either via the mathematical structure of trumping~\cite{Daftuar01_trumping, Klimesh07_ineqs}, majorization~\cite{Anspach01_condition}, or resource monotones~\cite{Brandao15_2ndlaws, vdMeer17_smoothed, Anshu18_general, Rethinasamy20_RelEnt}.
These approaches are powerful because they guarantee the existence of a catalytic operation given initial and target states. 
Nevertheless, most of them do not provide insight into the required catalyst state, or the process necessary for achieving the transformation.
This knowledge gap poses an obstacle towards practical demonstrations of catalytic processes. 
A handful of recent advancements have been made to address these concerns. 
For example, a significant finding is that given any catalytically possible state transition, nearly any quantum state can serve as a catalyst, as long as many copies are used~\cite{Lipka-Bartosik_PRX_universal}. This insight comes from the reversible convertibility in the i.i.d. limit, combined with the fact that catalyst states are not altered by the activation process. While this result offers incredible insight towards the catalytic power of i.i.d. states, the construction in~\cite{Lipka-Bartosik_PRX_universal} always prescribes the usage of a very high-dimensional catalyst, which is excessive for processes that might require only a small catalyst and simple operations\footnote{See~\cite{Gupta22_strontcat} for an alternative direction discussing catalysts that activate the transitions that are not possible even between multiple copies of initial and target states.}. These factors reduce the implementability of catalytic processes and highlight the need for a complementary study that emphasizes the extent to which small catalysts and straightforward operations can be effective.

Another challenge towards practical demonstrations of catalytic thermal processes comes from the genericity of thermal operations (TO)~\cite{Janzing00_TO} in the resource-theoretic setting~\cite{Brandao13_TRT, Horodecki13_fundamental, Gour15_NOreview, Brandao15_2ndlaws, Ng18_Qthermo_book}. TOs potentially involve intricate control over the global system and bath, in particular its joint energy eigenstates. Interestingly, crude operations~\cite{Perry18_PRX_Crude} acting on a  system plus a qubit Gibbs ancilla were proven to be sufficient to generate all TO state transitions to incoherent final states. Achieving target states with higher resourcefulness, however, can still be demanding, e.g. requiring the extreme lowering/raising of system and bath composite energy level, or arbitrary energy preserving unitaries in the system subspace. Therefore, the full employment of resource theoretic models in real world remains challenging~\cite{YungerHalpern17_Realization}.

Elementary thermal operations (ETO)~\cite{Lostaglio_18_ETO} alleviate the above issues by considering a subset of TO that can be decomposed into series of two-level swaps. This decomposition offers a natural way to prescribe a process to the experimenter, analogous to how a complex $ n $-qubit computation is decomposed into a small gate set\footnote{This analogy is not exact because, while in quantum computing, universal 2-qubit gate sets exist, it is known that TO cannot always be decomposed into ETO~\cite{Lostaglio_18_ETO}.}. Moreover, physical models, such as the collision model~\cite{Rau63_collision} or the Jaynes-Cummings (JC) model~\cite{JCmodel} can emulate the two-level swap. The intensity-dependent Jaynes-Cummings model~\cite{Buzek89_intensityJC} can furthermore generate all energy-preserving two-level swaps, making the setup more realistic.  Another advantage of ETO is that it opens up the opportunity to analyse intermediate states, which are found by partially applying the swap sequence -- providing a time-resolved description of the system dynamics, rarely possible in resource-theoretic thermodynamics. Nonetheless, ETO has its limitations: the only existing method to decide the feasibility of state transitions is to iteratively find all extreme points of the set of reachable states. The number of iterations grows extremely quickly with the system dimension.

Our aim is to leverage the operational simplicity of ETO to find examples of easy-to-realise catalytic evolutions. This goal is nontrivial since the conditions for ETO transitions, even without catalysts, are not well-characterised. Similar to the case of the resource theory of magic~\cite{campbell2011catalysis}, our result offers a straightforward recipe for implementing catalytic protocols. This strengthens the connection between catalytic studies in quantum thermodynamics and physically motivated, relatively practical thermodynamic operations.

In Sec.~\ref{sec:ETO_overview}, we provide an overview of TO and ETO to set the foundation for our study. We expound on the concept of tight-majorization and neighbouring $\beta$-swaps, elucidating them with Figs.~\ref{fig:thermo_curves} and~\ref{fig:ETO_freeEs}. Additionally, we develop a fundamental technical tool, i.e. Lemma~\ref{lemma:neighbouring_TOext} which plays a pivotal role in the development of our primary findings. In Sec.~\ref{sec:CETO}, we present our initial examples of catalytic advantages in ETO and display the complete set of reachable states facilitated by any qubit catalyst. The catalytic regime that we study possesses several desirable properties: it can be entirely achieved through two-level swaps, both system and catalyst are small in dimension, and the catalyst is fully recovered without any errors or remnant correlations with the system. 

Sec.~\ref{sec:higherd_cat} extends our investigation to a higher-dimensional regime, pushing the total system dimension from six, as explored in Sec.~\ref{sec:CETO}, up to ninety. In general, the case of $d=9$ already exceeds our computational capability. Nevertheless, by imposing restrictions on our initial states, we develop a significantly more efficient algorithm to determine the feasibility of (catalytic) ETO transitions. We focus our attention on the problem of cooling one thermal state to another under catalytic ETO and compare its performance to non-catalytic ETO and TO. Our results underscore that even relatively small catalysts suffice to overcome the TO limit.

In Sec.~\ref{sec:characterising_Seto}, we present the analytical results that form the basis for our numerical methods in Secs.~\ref{sec:CETO} and~\ref{sec:higherd_cat}. This includes, for example, a full characterisation of the set of states reachable by ETO for systems of dimension $d=3$ (Thm.~\ref{thm:3dETO}), as well as a universal upper bound, tighter than previous bounds (Thm.~\ref{thm:ETO_cone}). Notably, we identify a class of initial states where characterising the set of ETO-reachable states is computationally as efficient as checking thermo-majorization relations (Thm.~\ref{thm:nice_order}). This result essentially resolves the characterisation challenge of ETO for specific sets of initial states, opening the door for systematic investigations into operationally important tasks, such as the cooling scenario in Sec.~\ref{sec:higherd_cat}.

\section{Background}\label{sec:ETO_overview}

In this section we define the terminology and notation used throughout the paper, while also providing an overview of the state-of-art knowledge and techniques which we applied to develop our results. Our approach in this work is grounded in the resource-theoretic framework of quantum thermodynamics, which explores the fundamental constraints on thermodynamic state transitions concerning a set of \emph{free operations} and \emph{free states}. This approach has been fruitful~\cite{Lostaglio15_coherence, Cwiklinski15_PRL, Faist15_GPmap, Korzekwa16_work, Alhambra16_flucwork, Masanes17_3rdlaw, Wilming17_3rdlaw, Ng17_Carnot, Chubb18_finitesize, Korzekwa19_Irr, Woods19_efficiency, YungerHalpern20_photoisomerization, Shiraishi20_d-maj} in identifying the additional restrictions that arise due to non-negligible fluctuations and finite size effects for quantum systems that interact with their respective environments.

\subsection{Thermal operations}
The most well-known and studied version of the thermodynamic resource theory is the paradigm of thermal operations (TO).
These are quantum channels that can be written in the form 
\begin{equation}
\Eps(\rho_{S}) = \Tr_{B'}\left[U_{SB}\left(\rho_{S}\otimes\tbeta_B\right) U_{SB}^{\dagger}\right],\label{eq:TO}
\end{equation}
where 
\begin{enumerate}
	\item $\tbeta_B = \frac{\exp\left[-\beta H_B\right]}{Z_B}$ is the Gibbs state at temperature $1/\beta$ given an arbitrary choice of a bath Hamiltonian $H_B$, 
	\item $Z_B = \Tr\exp\left[-\beta H_B\right]$ is the partition function of the bath, 
	\item $U_{SB}$ is a global energy-preserving unitary, and therefore satisfies $[H_S+H_B,U_{SB}]=0$, and
	\item $ B' $ is a subsystem of choice within $ SB $ which is discarded at the end of the process\footnote{For most cases, we are interested in $B=B^\prime$.}.
\end{enumerate}

The complete state transition conditions of determining whether an initial $ \rho_S $ can be transformed into some $ \rho_S' $ via $ \Eps $ have remained a long-standing open problem, in particular when both initial and final states are energy-coherent. However, the situation simplifies in the quasi-classical case, where either $ \rho_S $ or $ \rho_S' $ is energy incoherent. Given a system $ S $, let us describe any given initial state
	$ \rho_S = \sum_{ij} a_{ij} \ketbra{E_i}{E_j} $,
with respect to the acsendingly-ordered energy eigenbasis, namely $ E_i \leq E_j $ for all $ i \leq j $. In the quasi-classical regime, the state is fully characterised by the population vector for each energy eigenbasis, which we denote by $ \pstate = \lbrace a_{ii} \rbrace_i $. We therefore in the rest of the manuscript refer to the population vectors as states living in some probability space $ \probspace^d $. The set of reachable final states from $ \pstate $ is then denoted as $ \clS_\mathrm{TO}(\pstate) $. If $ \qstate \in \clS_\mathrm{TO}(\pstate) $, then
	 $ \rho_S~ \xrightarrow{\rm TO}~\rho_S' = \sum_{i} q_i \ketbra{E_i}{E_i} $,
that is, $ \rho_S' $ is achievable from $ \rho_S $ via thermal operations.

For quasi-classical states, the possibility of state transition is fully characterised by \emph{thermomajorisation relations}~\cite{Horodecki13_fundamental}. 
It is furthermore known that the action of TO on such quasi-classical states can equivalently be described by Gibbs preserving matrices $G$ acting on the population vectors~\cite{Janzing00_TO, Horodecki13_fundamental}. That is, there exists $ G $ that preserves the Gibbs state population vector $\tbeta = \{e^{-\beta E_i}/Z_S\}_i$, i.e. $G\tbeta = \tbeta$, and $\qstate = G\pstate$.

Given two states $ \rho,\rho' $ represented by their population vectors $\pstate, \qstate$, and the system Hamiltonian (which in turn determines $ \tbeta $), thermomajorisation realtion, denoted by $ \succ_\beta $, is a pre-order between the states, defined according to their respective thermomajorisation curves $ \mathcal{L}_\pstate,\mathcal{L}_\qstate $ (Def.~\ref{def:therm_curve}). We say that $\pstate\succ_\beta \qstate$ if $\mathcal{L}_\pstate(a)\geq\mathcal{L}_\qstate(a)$ for all $ a\in [0,1] $.

%

\begin{definition}[$\beta$-order]\label{def:border}
	Given $\pstate\in\probspace^d$ and a Gibbs state $\tbeta\in\probspace^d$, let us denote the element-wise ratio of the two vectors as 
	\begin{equation}\label{eq:slopes}
		\slope(\pstate)_i := p_{i}/\tbeta_{i}.
	\end{equation}
	Then the $\beta$-order $\border(\pstate,\tbeta) $ is a particular ordering of the energy eigenbasis labels $ (1,\cdots,d) $, such that the ratios according to this ordering is non-increasing, i.e. 
	\begin{equation}\label{key}
		\slope(\pstate)_{ \pi_k} \geq  \slope(\pstate)_{ \pi_{k+1}},\quad \forall k\leq d-1,
	\end{equation}
	with $ \pi_k = \border(\pstate)_k$, where we omit $\tbeta$ in the argument when it is obvious from context.
\end{definition}

\begin{definition}[Thermo-majorization curve]\label{def:therm_curve}
	For a state $\pstate\in\probspace^d$, the themo-majorization curve $\mathcal{L}_\pstate:\lbrace 0,1\rbrace\rightarrow \lbrace 0,1\rbrace $ is a piecewise-linear function that interpolates between the coordinates $\{(0,0)\}$ and elbow points $\{(\sum_{k=1}^{l}\tbeta_{ \pi_k}, \sum_{k=1}^{l}p_{ \pi_k})\}_{l=1}^{d}$.
\end{definition}

Earlier works have shown that the set $  \clSto(\pstate) $ characterised by thermomajorisation can be constructed quite efficiently. This is summarised in the theorem below. 

\begin{theorem}[\cite{Lostaglio_18_ETO}, Lemma~12 and \cite{Mazurek19_channels}, Thm.~2]\label{thm:TO_cone} 
	For any given $ \pstate $, the set of reachable states $\clSto(\pstate)$ is a convex combination of $ d! $ unique extreme points that correspond to distinct $\beta$-orders and are tightly thermomajorised by $ \pstate $.
\end{theorem}

With Thm.~\ref{thm:TO_cone}, determining $\clSto(\pstate)$ is computationally inexpensive. This theorem follows from the fact that if a state is tightly-majorised by $ \pstate $, then it thermomajorises any other state in $ \clSeto(\pstate) $ that has the same $ \beta $-order $\border^{(m)}$. The importance of tightly-majorised states can also be seen from another perspective:

\begin{lemma}[Thm.~12 of~\cite{Perry18_PRX_Crude}]\label{lemma:same_border}
	If two states $\pstate$ and $\qstate$ have the same $\beta$-order $\border(\pstate) = \border(\qstate)$ and $\pstate\succ_{\beta}\qstate$, then $\qstate$ can be obtained from $\pstate$ by a sequence of partial level thermalizations.
\end{lemma}

Partial level thermalizations~\cite{Perry18_PRX_Crude}, where populations of a subset of levels are partially mixed with corresponding thermal populations, can always be decomposed into series of two-level partial swaps and thus are (E)TO.

\begin{figure}
	\centering
	\includegraphics[width=0.44\linewidth]{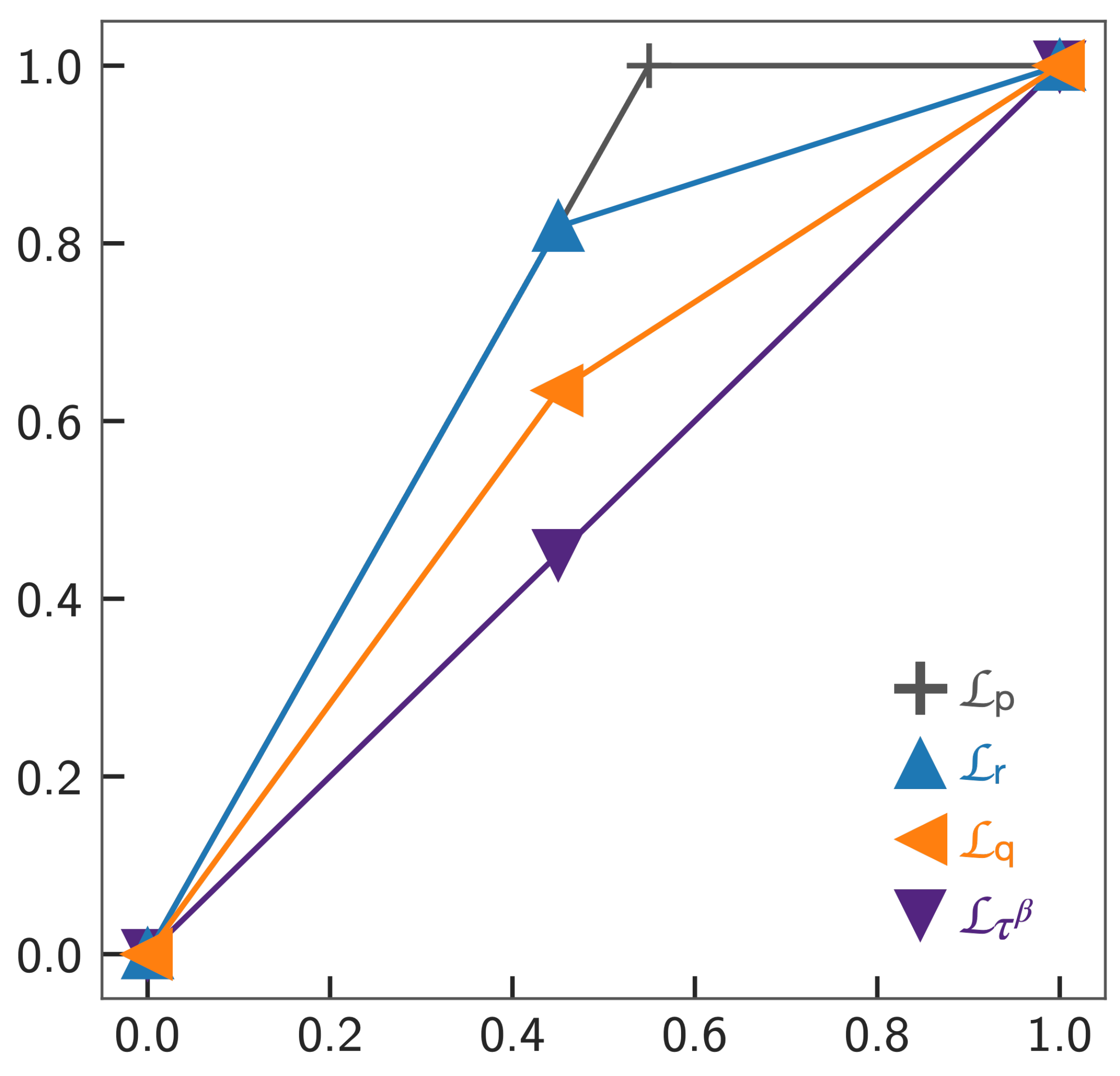}
	\caption{Thermomajorization curves of $ \pstate,\rstate,\qstate $, and $ \tbeta $, where $ \pstate\succ_{\beta}\rstate\succ_{\beta}\qstate\succ_{\beta}\tbeta $. In particular, $ \rstate $ is tightly thermomajorised by $ \pstate $, implying that $ \rstate $ thermomajorises any $ \qstate\in\clSeto(\pstate) $ with the same order. Here we plotted one such state $ \qstate = (\rstate+\tbeta)/2 $. Temperature and energy levels are set to be $E_1 = 0$, $ \beta E_2 = 0.2$.}
	\label{Fig:thermomaj_example}
\end{figure}

\begin{tcolorbox}[breakable, every float=\centering, drop shadow, title=An example of tight thermomajorisation]
				Given a qubit Hamiltonian ${\rm H} = (E_1,E_2)$ with its corresponding Gibbs state $\tbeta = (\tbeta_1 , \tbeta_2)$ of some fixed temperature $1/\beta$, consider the pure ground state
					$ \pstate = (1,0) $,
			which has a $\beta$-order $\border(\pstate) = (1,2)$ and a simple thermomajorisation curve 
\begin{align}
	\mathcal{L}_\pstate(a) = \begin{cases}
		a(\tbeta_1)^{-1}, &\text{for } a\leq \tbeta_1,\\
		1, &\text{for }  \tbeta_1<a\leq 1.
	\end{cases}
\end{align}
A state $\qstate = (q_1,q_2)$ has $\border(\qstate) = (2,1)$ if $q_2>\tbeta_2$, which leads to a thermomajorisation curve
\begin{align}
	\mathcal{L}_\qstate(a) = \begin{cases}
		aq_2(\tbeta_2)^{-1}, &\text{for } a\leq \tbeta_2,\\
		1 - \frac{(1-a)q_1}{1-\tbeta_2}, &\text{for }  \tbeta_2<a\leq 1.
	\end{cases}
\end{align}

Then $\pstate\succ_\beta\qstate$ if $q_2\leq \exp[\beta(E_1-E_2)]:= \Delta_{12}$. Furthermore, when $r_2 = \Delta_{12}$, $\rstate = (r_1,r_2)$ is \emph{tightly thermomajorised} by $\pstate$, i.e. the elbows of $\mathcal{L}_\rstate$ coincide with the curve $\mathcal{L}_\qstate$. See Fig.~\ref{Fig:thermomaj_example} for an example of tight-thermomazorization.

\end{tcolorbox}

Another important aspect of state transition conditions comes from functions that behave monotonically whenever thermal operations are applied. The most commonly known set of such monotones is dubbed generalised free energies~\cite{Brandao15_2ndlaws} 
\begin{equation}\label{eq:Falpha}
	F_\alpha(\rho, \tbeta): = \frac{1}{\beta} \left(D_\alpha(\rho\Vert\tbeta) - \log Z\right),
\end{equation}
for all $\alpha\in(-\infty,\infty)$. 
The R\'enyi divergence for a quasi-classical state~\footnote{Multiple generalizations for energy-coherent states exist, and can be found in Ref.~\cite{Brandao15_2ndlaws}.} $\pstate\in\probspace^d$ reduces to
	$ D_\alpha (\pstate \Vert \tbeta) = \frac{\sgn(\alpha)}{\alpha-1}\log\sum_{i}p_i^\alpha (\tbeta_i)^{1-\alpha}$.
When $\alpha \rightarrow 1$, $F_\alpha$ reduces to the nonequilibrium free energy 
\begin{align}
	F(\rho,\tbeta) = \lim_{\alpha\rightarrow1^+} F_\alpha(\rho\Vert\tbeta)  = \langle H\rangle_\rho - \frac{1}{\beta}S(\rho),
\end{align}
where $\langle\cdot\rangle_\rho$ denotes average with respect to the state $\rho$, and $S(\rho)$ is the von Neumann entropy -- a recovery of the second law of thermodynamics. 

Monotones only provide necessary conditions for state transitions via TO and its subsets. Yet, the set of monotones in Eq.~\eqref{eq:Falpha} becomes sufficient in the quasi-classical regime~\cite{Brandao15_2ndlaws}, when catalysts are involved. In other words, it is straightforward to determine the existence of a catalytic thermal operation (CTO) that achieves $ \rho\xrightarrow{\rm CTO}~\rho^\prime $ -- one simply needs to check that Eq.~\eqref{eq:Falpha} is monotonic for all $ \alpha $ with respect to such a state transition.  

\begin{table*}[htbp]
	\centering
		\begin{tabular}{ |c||L{3.35cm}|L{3.42cm}|L{3.7cm}|}
			\toprule
			\textbf{Operations} & \multicolumn{3}{l}{\hspace{1.1cm}MTP \hspace{1cm} $\subsetneq$  \hspace{1.05cm} ETO  \hspace{1.05cm}  $\subsetneq$  \hspace{1.2cm} TO \hspace{2cm}}  \vline\\
			\midrule
			Construct Ext$[\clS]$ &  Neighbouring $\beta$-swap series (Thm.~4 of ~\cite{Lostaglio22_MTP1}) & $\beta$-swap series (Thm.~\ref{thm:ETO_cone}) & Tight thermomajorisation (Thm.~\ref{thm:TO_cone})\\
			\hline
			$\pstate\rightarrow\qstate$ iff. & $\pstate$ continuously thermomajorises $\qstate$ & Unknown  & $\pstate$ thermomajorises $\qstate$\\
			\hline
			Catalysis & Examples of Gibbs state catalysts~\cite{Korzekwa22_MTP2} & Qubit examples (Sec.~\ref{sec:CETO}); higher-dimensional examples (Sec.~\ref{sec:higherd_cat}) & $\exists$ catalytic transformation iff. $F_{\alpha}$ does not increase $\forall\alpha\geq0$\\
			\hline
		\end{tabular}
		\caption{Comparison between Markovian thermal processes, elementary thermal operations, and thermal operations, when applied to the quasi-classical cases. \label{table:comparison}}
	\end{table*}

\subsection{Elementary thermal operations}
Special subsets of TO, which take into account more realistic limitations on feasibility, have been developed in the past few years. Two prominent examples are Markovian thermal processes (MTP)~\cite{Lostaglio22_MTP1}, which we will not explicitly describe in this paper, and elementary thermal operations (ETO)~\cite{Lostaglio_18_ETO}, which is the focus of this paper. See Table~\ref{table:comparison} for a hierarchy between these operations. 

ETO is a subset of TO with an additional restriction on the unitary transformation $U_{SB}$ in Eq.~\eqref{eq:TO}. Unlike TO, where the simultaneous manipulation of all energy levels of the system is allowed, ETO concatenates a series of operations where each individual step involves only two energy levels. 
Formally, 
\begin{equation}
	\langle i \rvert_{S} U_{SB} \lvert i \rangle_{S} = \mathbb{1}_B,\ \forall i\neq j,k,
\end{equation}
where $j$ and $k$ are the system energy levels that the unitary aims to maneuver. 
As in TO, the state transition from $\pstate$ to $\qstate$ via a single ETO step acting on two levels 
can be written as 
$ 	\qstate = \Mswap^{(j,k)}_\lambda\pstate $,
where~\footnote{Whenever the order of $E_j$ and $E_k$ is known, we denote ETO such as $ \Mswap^{(j,k)}$  or $\beta^{(j,k)} $, using the convention that $E_j\leq E_k$.} \begin{align}\label{eq:Deltas}
	\lambda\in[0,1], &\quad \Delta_{jk} := \exp[\beta(E_j-E_k)],
	\end{align}
and
\begin{equation}
	\Mswap^{(j,k)}_{\lambda}: = \begin{pmatrix}
		1-\lambda \Delta_{jk} && \lambda \\
		\lambda \Delta_{jk} && 1-\lambda
	\end{pmatrix} \oplus \mathbb{1}_{\setminus (j,k)}.\label{eq:Mswap_def}
\end{equation}
This channel is also constructed to preserve the Gibbs state $\tbeta$. 
Extremal cases of $\lambda=1$ are named $\beta$-swaps (throughout the manuscript we also refer to them simply as swaps),
\begin{equation}\label{eq:beta_jk}
	\beta^{(j,k)}:= \Mswap^{(j,k)}_1 = \begin{pmatrix}
	1- \Delta_{jk} & 1 \\
		\Delta_{jk} & 0
	\end{pmatrix} \oplus \mathbb{1}_{\setminus (j,k)}.
\end{equation} 

By definition, it is clear that both \emph{concatenations} and \emph{convex combinations} of TOs are themselves also TO. This is obviously not true for ETO, whenever they act non-trivially over different energy levels. 
Therefore, most of the interesting transitions would arise only when we include all sequences of ETOs, and also arbitrary convex combinations as allowed operations.
Such processes are always TO, hence, 
\begin{equation}\label{key}
	\rho \xrightarrow{\rm ETO}~\rho' \quad\implies\quad \rho\succ_\beta\rho',
\end{equation}
but the converse is not true
~\cite{Lostaglio_18_ETO, Mazurek18_decomp}. 

Since no efficient way to determine the possibility of $ \pstate $ to $ \qstate $ via ETO is known, one must construct $\clSeto(\pstate)$ starting from the initial state $\pstate$, and then verify whether $ \qstate\in  \clSeto(\pstate)$. However, existing theoretical constructions of $  \clSeto(\pstate)$ (by its extreme points) grow rapidly in the system dimension, and become computationally infeasible for the simplest examples of catalysis. We partially address this issue with our results in Sec.~\ref{sec:characterising_Seto}.

\subsection{Neighbouring swaps}\label{subsec:neighbouring}

To find the extreme points of $ \clSeto(\pstate) $, \emph{neighbouring swaps} are employed -- these are $\beta$-swaps of the form $\beta^{(\pi_j,\pi_{j+1})}$ and applied to a state $\pstate$ with $\border(\pstate)_k = \pi_k$, i.e. they swap two consecutive levels in the $\beta$-order of the input state. Neighbouring swaps are important because they typically minimise dissipation/change in athermality (as captured by thermomajorisation, see Lemma~\ref{lemma:neighbouring_TOext}).
\begin{lemma}~\label{lemma:neighbouring_TOext}
	For any state $ \pstate $, the state obtained from a neighbouring $\beta$-swap $\qstate^{(j)} = \beta^{(\pi_j, \pi_{j+1})}\pstate$ is tightly thermomajorised by $\pstate$.  Furthermore, states obtained by a single but non-neighbouring $ \beta $-swap are never tightly-majorised by $\pstate$ unless two swapped levels have the same energy. 
\end{lemma}

\begin{figure}[t!]
	\includegraphics[width=\columnwidth]{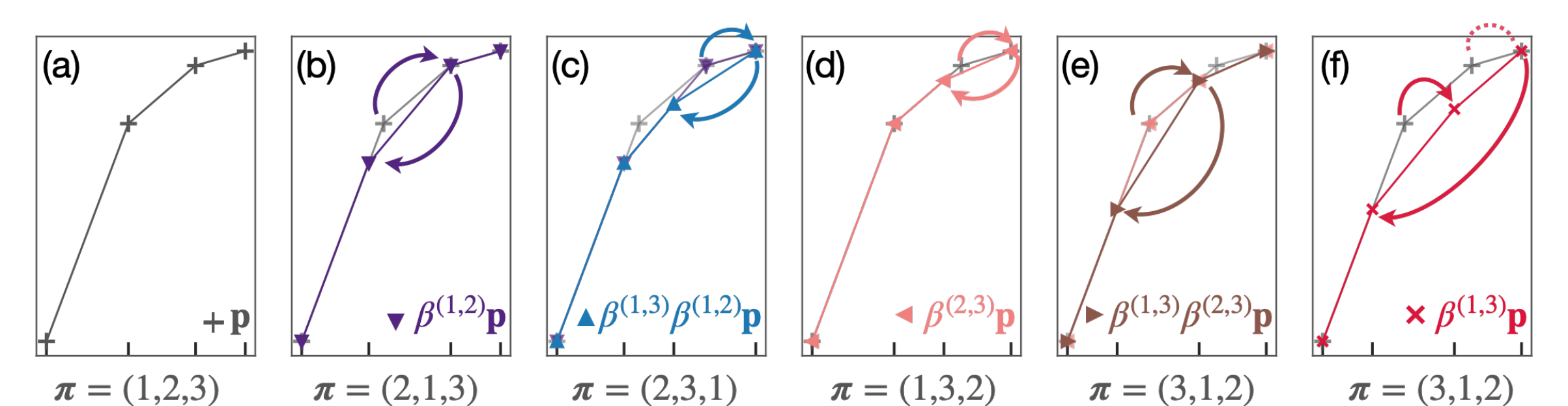}
	\caption{Thermo-majorization curves of a state undergoing $ \beta $-swaps. The initial state population is set to $\pstate = (0.75,0.2,0.05)$ and the Hamiltonian is set to $\beta\rm{H}_S = (0,0.2,0.5)$. 
		Panel {(a)} shows a three-dimensional initial state $\pstate$. Panels (b) and (d) each shows the state after a single neighbouring swap, while {(f)} shows a single non-neighbouring $\beta$-swap. One observes that both (b) and (d) thermomajorise (f). Furthermore, panels (c) and (e) show $ \pstate $ after two consecutive neighbouring $\beta$-swaps. Observe that even (c) and (e) thermomajorise (f), since the non-neighbouring swap produced a much larger deviation from the initial thermomajorisation curve. 
	}
	\label{fig:thermo_curves}
\end{figure}

Lemma \ref{lemma:neighbouring_TOext} is proven in \ref{app:proof_lem3}, and furthermore, one can see examples of the lemma visualised in Fig. \ref{fig:thermo_curves} (b), (d) and (f). 
From Thm.~\ref{thm:TO_cone}, it follows that the neighbouring swapped states $\qstate^{(j)}$ are identified as extreme points of $\clSto(\pstate)$, and therefore thermomajorise all states $\rstate\in\clSto(\pstate)$ such that $\border(\qstate) = \border(\rstate)$. This leads to the following corollary:

\begin{corollary}\label{corollary:TO-ETO_ext}
If a state $\qstate \in \clSeto (\pstate)$ is tightly thermomajorised by $ p_0 $, then it is a unique extreme point of $\clSeto (\pstate)$ among states of the same $\beta$-order.
\end{corollary}

Therefore, all states $\qstate^{(j)}$ obtained from a single neighbouring swap are also unique extreme points of $\clSeto(\pstate)$ having order $\border(\qstate^{(j)})$. Numerical examples such as Fig.~\ref{fig:thermo_curves} (c) show that when multiple neighbouring swaps are applied, the final state is no longer necessarily extremal. Nevertheless, most of the extreme points of $\clSeto(\pstate)$ are of the form $\qstate^{(\vec{j})}$~\footnote{This is a state after $l$ neighbouring swaps. We use the notation $\qstate^{(\vec{j},j_{l+1})} = \beta^{(\border(\qstate^{(\vec{j})})_{j_{l+1}}, \border(\qstate^{(\vec{j})})_{j_{l+1}+1})}\qstate^{(\vec{j})}$, where $\vec{j} = (j_1,\cdots,j_{l})$}, or contain only a small number of non-neighbouring swaps in their construction. 

Fig.~\ref{fig:ETO_freeEs} exemplifies this preference towards neighbouring swaps -- there, the free energy of the state is still higher after three neighbouring swaps, compared to a single non-neighbouring swap. 
In (a), when a neighbouring $\beta$-swap (blue arrow) is applied, the state moves to the adjacent cell of $\beta$-order, whereas the non-neighbouring swap (red arrow) transfers a state through multiple boundaries of different orders (black dashed lines), resulting in a state less resourceful than $\qstate^{(2,1,2)} = \beta^{(1,2)}\beta^{(2,3)}\beta^{(1,3)}\pstate_1$, that experienced three neighbouring swaps. 
Fast exhaustion of the athermality also occurs in different initial states, rendering non-neighbouring swaps typically suboptimal.   

\begin{figure}[t!]
	\includegraphics[width=0.92\columnwidth]{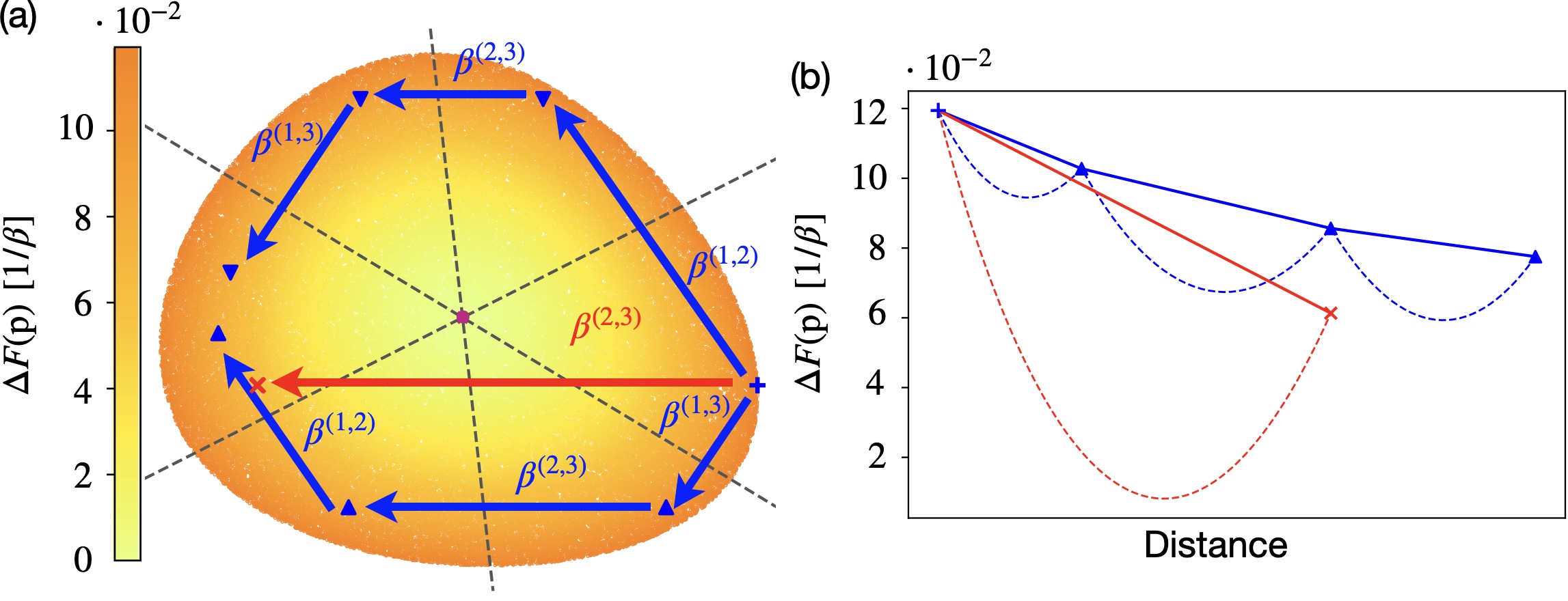}
	\caption{Initial $ d=3 $ state $\pstate_1 = (0.35,0.55,0.1)$ (blue plus marker) undergoing ETO, where its Hamiltonian is fixed as $\beta\mathrm{H}_S = (0,0.2,0.5)$, and all energies are scaled in a unit of $1/\beta$.
		Panel (a): ETO operations in barycentric coordinate, where only relevant states $ \qstate $ with $F (\qstate,\tbeta) \leq F (\pstate_1,\tbeta)$ are marked in colour.
		Triangles label the extreme points of $\clSeto(\pstate_1)$, achieved by $ \beta $-swaps indicated by blue arrows. The red X marker is the state $\beta^{(2,3)}\pstate_1$.
		Different $\beta$-order cells are separated by black dashed lines connecting pure states and the thermal state. 
		Panel (b): the free energy differences from the equilibrium state $\tbeta$, in two different paths corresponding to (a). 
		Dashed lines show the continuous values of $ \Delta F $ from points on arrow paths in (a). Straight lines connect the discrete values obtained from endpoints denoted with triangle and X symbols. 
		$ x $-axis is the total length of the path taken from the initial state as plotted in (a), e.g. $ x $-coordinate of the second triangle is the summation of the first and the second blue arrow lengths starting from the plus symbol.   
	} 
	\label{fig:ETO_freeEs}
\end{figure}

Moreover, it is useful to note that Fig.~\ref{fig:ETO_freeEs} displays the real time evolution of the state during ETO transitions. By construction, only two levels of the state undergo change in time, i.e. all the other populations are fixed during that period, imposing the system to follow the straight lines in barycentric representations as in (a). In (b), dashed lines correspond to the free energies during the real time evolution when the intensity-dependent Jaynes-Cummings model is assumed (See \ref{appendix:dynamics} for the real time reduced state dynamics). Suppose we apply $ \beta^{(j,k)} $ to a state $ \pstate $ with $\slope(\pstate)_j>\slope(\pstate)_k$. According to Eqs.~\eqref{eq:conti_evolution1} and~\eqref{eq:conti_evolution2}, the state evolves as $ \qstate(\lambda(t)) $, where $ \qstate(\lambda = 1) = \beta^{(j,k)}\pstate $ with $ \slope(\qstate)_{k}>\slope(\qstate)_{j} $. 
Since the evolution is continuous, there exists a state $ \qstate(\lambda^{*}) $, such that $ \slope(\qstate)_{j} = \slope(\qstate)_{k} $. This point is given by $\lambda^* = 1/(1+\Delta_{jk})$ ($\lambda^* = 1/(1+\Delta_{kj})$) when $k\geq j$ ($k\leq j$). The state $ \qstate(\lambda^{*}) $ is the closest to being thermal, and achieves minimal generalised free energies -- minima of dashed lines in (b) -- among states $ \qstate(\lambda(t)) $.
In other words, during a single $\beta$-swap, free energy decreases until the minimum point $\qstate_{\lambda^*}$ is reached, and then increases until the end of the $ \beta $-swap. On the other hand, if only the endpoints values are considered (solid lines of (b)), free energy cannot increase after each ETO step. These intermediate increases within a single swap evince the non-Markovian effect of thermal reservoirs at each step, differentiating ETO with strictly Markovian thermal processes~\cite{Lostaglio22_MTP1, Korzekwa22_MTP2}. See Sec.~\ref{sec:discussion} for more discussion.

\section{Qubit catalysts for qutrit systems}\label{sec:CETO}

In this section we demonstrate the following:
\begin{enumerate}[label=(\roman*)]
	\item Existence of the catalytic advantage in the ETO framework. In particular, we report the discovery of simple qubit catalysts that are sufficient to produce a non-trivial advantage for state transitions in a qutrit system.
	\item A time-step resolved tracking of the dynamical path for system and catalyst, that reveals snapshots of the inner workings during the catalytic process.
\end{enumerate}
These results could be obtained only by a more systematic understanding of the set of reachable states via ETO, which are detailed in Sec.~\ref{sec:characterising_Seto}, where we establish various simplifications that enhanced the computability of $ \clSeto $. 

Throughout the rest of the manuscript, we use the vectors $\pstate, \qstate, \rstate$ to denote system states and $\catstate$ is reserved for catalyst states.

\subsection{Catalytic advantages}
For a qutrit system and qubit catalyst, the composite state $\pstate\otimes\catstate$ lives in a six-dimensional probability space $ \probspace^6 $. Our goal is to construct a set obtainable by catalytic elementary thermal operations with a qubit catalyst (CETO$_2$)
 \begin{equation}\label{key}
	\clScetoqb(\pstate) = \lbrace \qstate| \exists \catstate\in\probspace^2~{\rm s.t.}~\pstate\otimes\catstate\xrightarrow{\rm ETO}~\qstate\otimes\catstate  \rbrace.
\end{equation}
By definition, $\clSeto(\pstate)\subset \clScetoqb(\pstate) $, but qubit catalytic advantage exists iff. $\clScetoqb(\pstate)\not\subset\clSeto(\pstate)$. 
For a number of limited cases, 
given a fixed catalyst state $\catstate = (c_1,1-c_1)$, some parts of the set
\begin{equation}\label{key}
	\clSceto(\pstate;\catstate) = \lbrace \qstate| \pstate\otimes\catstate\xrightarrow{\rm ETO}~\qstate\otimes\catstate  \rbrace
\end{equation}
can be evaluated analytically (\ref{appendix:CETO}). 
Nonetheless, $ \clSceto(\pstate;\catstate) $
is in general constructed by numerically evaluating the extreme points of $\clSeto(\pstate\otimes\catstate)$ and imposing exact catalyst recovery conditions. The set $\clScetoqb(\pstate)$ is then given by iterating the process for different values of $c_1$. In this first work, we focus on the case where $ S $ remains uncorrelated to $ C $ to affirm that catalytic advantages exist even in the most conservative catalytic setting. We expect the catalytic power to further increase when system-catalyst correlation persists as has been shown in different resource theories~\cite{Muller18_corr, Boes19_catconjecture, Ding21_corr, Shiraishi21_PRL, Wilming21_catconjecture,Rubboli22_CorrLim,  Yadin22_Gaussian}, but leave this case for future work.

\begin{figure}[t!]
	\centering
	\includegraphics[width=0.9\columnwidth]{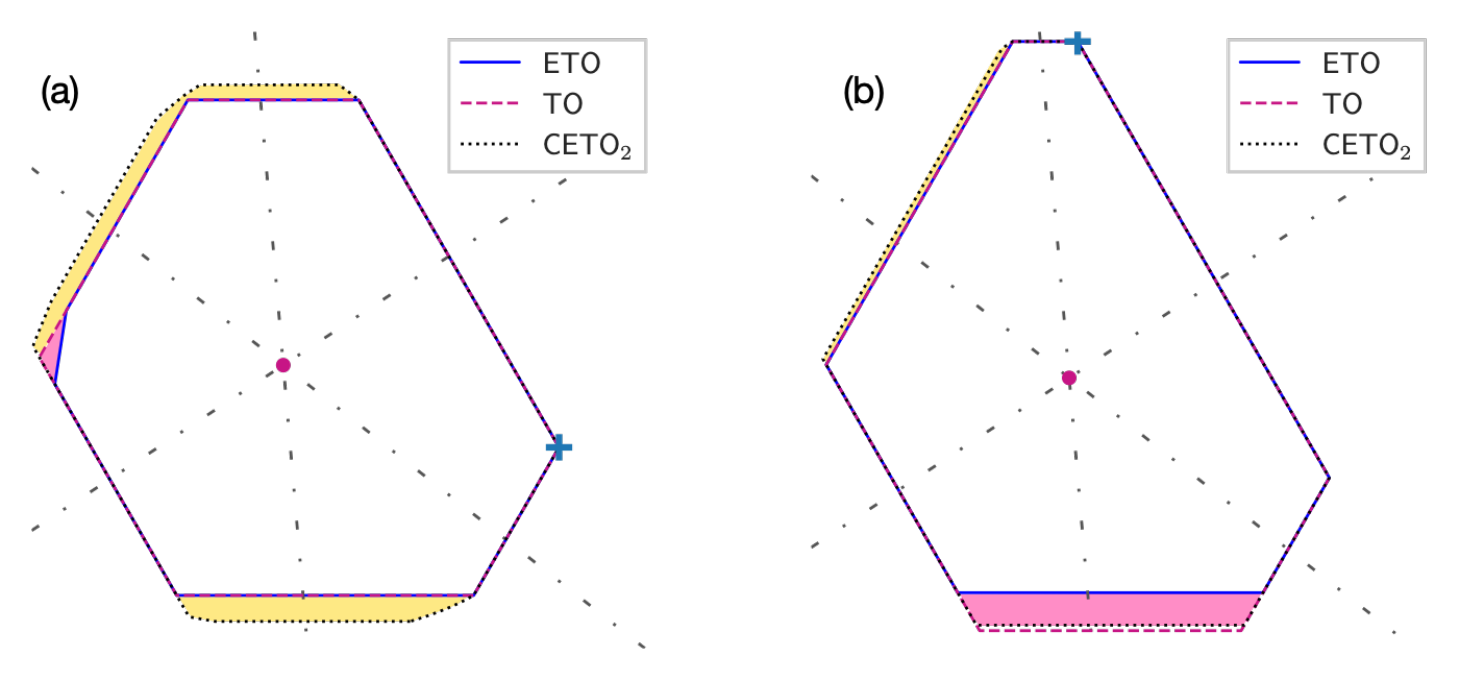}
	\caption{Barycentric representation of $\clSeto$ (blue solid lines), $\clSto$ (violet dashed lines), and $\clScetoqb$ (black dotted lines), for the initial states (a) $ \pstate_1 $ in Fig.~\ref{fig:ETO_freeEs} and (b) $\pstate_2 = (0.7,0.2,0.1)$, such that $\border(\pstate_2) = (1,2,3)$. The initial state, thermal state $ \tbeta $, and $ \beta $-order regimes are denoted similarly as to Fig.~\ref{fig:ETO_freeEs}. Pink shaded areas mark catalytic advantage within $\clSto$, while yellow shaded areas are states that can be reached via ETO+qubit catalyst, but not by TO without catalysis.
	}
	\label{fig:ext_points_213}
\end{figure}


In Fig.~\ref{fig:ext_points_213}, we present two $\clScetoqb$ sets corresponding to two different initial states that have distinct $\beta$-orders. The sets are displayed in comparison with the non-catalytic set of reachable states $\clSeto$ and $\clSto$. 
In particular, notice that in Fig.~\ref{fig:ext_points_213} (a), where the initial state is set to have $\border = (2,1,3)$, we observe that 
$ \clSto\subsetneq\clScetoqb $.
This feature persists for a number of randomly chosen initial states with $\border = (2,1,3)$ and $(3,1,2)$.
When this happens, we can on one hand reproduce every TO transition using ETO with a single qubit catalyst; and on the other hand, combat some of the finite-size effects and enable a larger set of transitions than previously allowed by an arbitrary TO. 
Likewise, consider Fig.~\ref{fig:ext_points_213} (b), for an initial state $\border = (1,2,3)$. Here, we observe that $\clScetoqb$ overlaps almost entirely with $\clSto$, but neither is fully contained by the other. This qualitative characteristic is again present for different initial states with the same $ \beta $-ordering.


In Fig.~\ref{fig:ext_points_213}, the set of states that go beyond $\clSto$ (yellow) highlights the additional advantage brought forward by CETO$_2$. In addition, the set of states between $\clSto$ and $\clSeto$ (pink) also has operational merits -- 
there exist states which require genuine multi-level TO to achieve, but can be obtained by an alternative pathway that involves only basic, JC-like interactions when a catalyst is present.
A few natural questions emerge: for one, it would be interesting to see how the gap between CETO$_n$ and ETO changes with $n$ being the dimension of the allowed catalyst. In Sec.~\ref{sec:higherd_cat}, some hints for this question is given. A second question would be how would CETO$_n$ fare when compared to CTO, i.e. the \emph{catalytic} version of thermal operations. The question of a generally fixed $n$ is difficult to answer, however, in a follow-up work, we proved that the set of energy-incoherent state transitions reachable by CETO converges to that of CTO, for the special case when $n$ is allowed to be arbitrarily large. In other words, when any catalyst is allowed, catalytic elementary thermal operations are as powerful as catalytic thermal operations for incoherent inital states~\cite{Son23_hierarchy}.

\subsection{Tracking catalytic processes}\label{subsec:tracking}

\begin{figure}[t!]
	\includegraphics[width=\columnwidth]{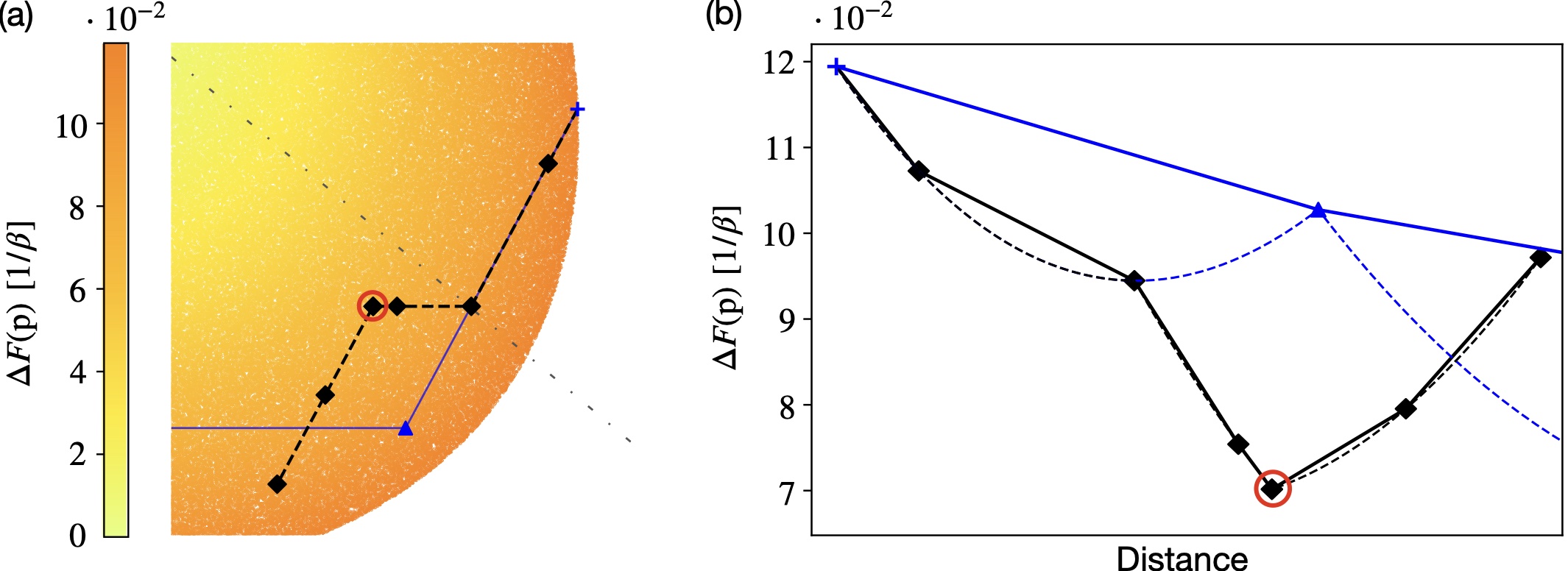}
	\caption{(a): Illustration of system reduced state evolutions in a CETO process (diamond markers and dashed lines) starting from an intial state $\pstate_1$ with final state lying beyond $\clSeto(\pstate_1)$. This state is chosen to be an extreme point of $\clScetoqb(\pstate_1)$. The rest of the plot is the zoomed-in view of Fig.~\ref{fig:ETO_freeEs} (a) around $\beta$-orders $(2,1,3)$ and $(2,3,1)$. 
	(b): Free energy differences from the equilibrium state $\tbeta$. Solid lines connect the values after each swap and dahsed lines mark continuous change between them.
	The fifth swap applied during the catalytic evolution does not alter the system reduced state, making the point before and after that swap not distinguishable in the system reduced picture of (a) and (b). These points are marked with red circles in both (a) and (b). Initial state $\pstate_1$ and its Hamiltonian $\beta\mathrm{H}_S$ are the same to those of Fig.~\ref{fig:ETO_freeEs}. 
}
	\label{fig:CETO_freeE}
\end{figure}

\begin{figure}[h!]
	\includegraphics[width=\columnwidth]{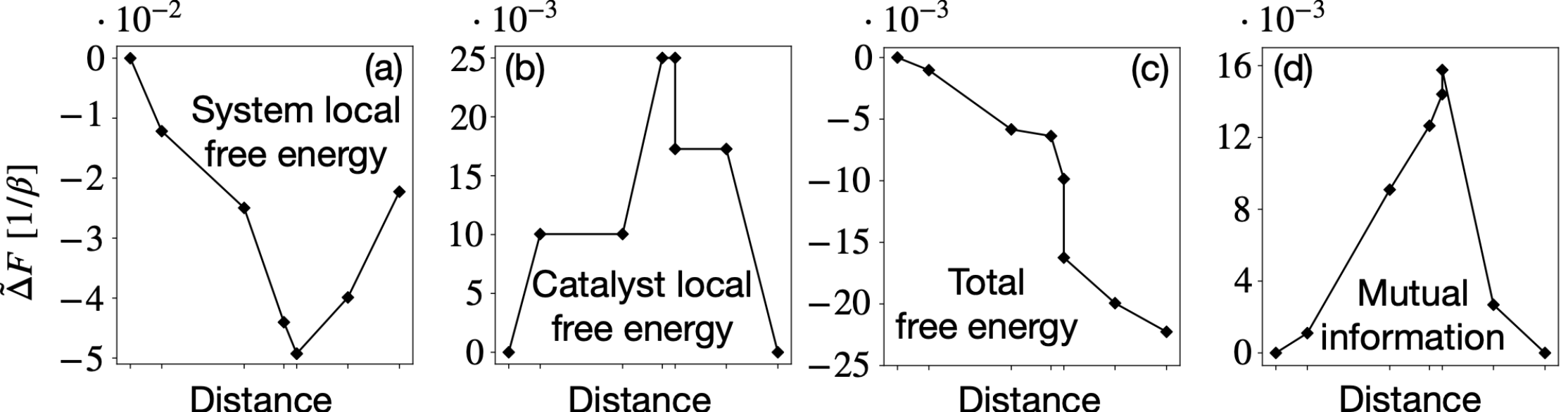}
	\caption{Nonequilibrium free energies of CETO evolution described in Fig.~\ref{fig:CETO_freeE}. The symbol $\tilde{\Delta}$ indicates the difference from the values of the initial state. (a): system local free energy identical to black lines in (b) of Fig.~\ref{fig:CETO_freeE}; (b): catalyst local free energy; (c): total free energy; (d): mutual information which, in this case, is identical to the difference between the total free energy and the sum of two local free energies. 
	$ x $-coordinates of points are identical to ones in Fig.~\ref{fig:CETO_freeE} (b).
	}
	\label{fig:local_freeEs}
\end{figure}

Here we explicitly construct a simple series of ETO swaps that leads to the transformation of an initial state $\pstate$ into an extreme point of $\clScetoqb(\pstate)$, and track changes in the system free energy throughout the process. 

In general, a state $\qstate$ being an extreme point of $\clSceto(\pstate)$ does not guarantee $\qstate\otimes\catstate$ to be an extreme point of the enlarged $\clSeto(\pstate\otimes\catstate)$ in system-catalyst composite space. Thus, (non-trivially) different paths of catalytic processes exist. 
In Fig.~\ref{fig:CETO_freeE}, we choose the shortest swap series, which comprises seven ETO steps, to realise the CETO transition from an initial state $\pstate_1$ to a new extreme state of $\clScetoqb(\pstate_1)$. Note that such a choice may not always be unique; see \ref{appendix:CETO} for details. Since adopting a qubit catalyst doubles the total dimension, and recovering the catalyst also requires additional swaps, more two-level swaps are needed to achieve a similar final state (as compared to a non-catalytic ETO).
For the first four steps, the nonequilibrium free energy of the system, $ F_S $ decreases; see Fig.~\ref{fig:CETO_freeE} (b). 
After two consecutive ETO steps on the joint system, the catalytic path seems identical to a non-catalytic ETO on the system. Nevertheless, on the global picture, correlations are already starting to build up with the catalyst, as shown in Fig.~\ref{fig:local_freeEs}.
In the third and fourth steps, $ F_S $ is further reduced until it has almost a half of the original nonequilibrium contribution. Recalling that free energies are the monotones of TO, this rapid depletion of the athermality would have been irreversible, if we have access only to the system. Yet, during catalysis, some of the free energy is transferred to either 1) the catalyst populations or 2) correlations between system and catalyst. This build-up is critical for the sixth and final steps of the process: stored free energy is what enables the system to again increase in local nonequilibrium free energy, therefore achieving a final state $ \qstate$ outside $\clSeto(\pstate_1)$ (and in this case, $ \qstate $ is even outside $\clSto(\pstate_1)$). See \ref{appendix:CETO} for a more detailed analysis. Similar behaviours are observed for other monotones, such as generalised free energies with $\alpha\neq 1$.  
Note that the most conservative form of catalysis is assumed: the catalyst is returned exactly and without correlations with the system. Hence, catalyst local free energy goes back to its original level, and mutual information also returns to zero, as shown in (b) and (d) of Fig.~\ref{fig:local_freeEs}. Their role is restricted to temporary free energy storages. Furthermore, the total free energy always decreases after each swap, which is expected, since the process is an ETO on the joint system.

\section{Higher-dimensional catalysts for transitions between thermal states}\label{sec:higherd_cat}

In general, identifying the whole set of reachable states $ \clScetoqd{d}(\pstate) $ with $ d $-dimensional catalysts is highly challenging, even when $ \pstate $ is three-dimensional and $ d = 3 $. However, by developing a theoretical tool, namely Thm.~\ref{thm:nice_order}, we show that the numerical cost is dramatically reduced for initial states $ \pstate $ whose $ \beta $-orders are monotonic in energy levels. A particularly important class of states that satisfy this property is the set of thermal states with different temperatures. When $ \beta_{h}<\beta $, i.e. the state $ \tau^{\beta_{h}} $ is hotter than the environment with temperature $ \beta^{-1} $, the $ \beta $-order $ \border(\tau^{\beta_{h}}) = (d,d-1,\cdots,1) $. Similarly, colder states with $ \beta_{c}>\beta $ have the order $ \border(\tau^{\beta_{c}}) = (1,2,\cdots,d) $. If we further employ a catalyst $ \catstate $ from the set of states which are sufficiently thermal, the monotonicity of $ \beta $-order would be preserved, i.e. $ \border(\pstate\otimes\catstate) $ is also monotonic in the total energy. In such cases, the analysis of higher-dimensional catalysts becomes computationally tractable. We will refer to such catalyst states as \emph{minimally-disturbing catalysts}, in the sense that they do not disturb the $\beta$-ordering of the system-catalyst composite.

To demonstrate the benefit of this reduction, let us consider the cooling process via (C)ETO starting from a high temperature thermal state $ \tau^{\beta_{h}} $ with $ \beta_{h}<\beta $ to a colder thermal state $ \tau^{\beta_{c}} $. For any temperature $ \beta_{h}^{-1} $, it is always possible to reach the ambient temperature $ \beta^{-1} $ by a full thermalization with the environment. However, with ETO and TO, colder temperatures $ \beta_{c}>\beta $ can be achieved, which has been studied in \cite{Scharlau2018quantumhornslemma} for the non-catalytic case of a qubit. 
One of our main theses is to showcase the effectiveness of small catalysts with practicable procedures. To corroborate this claim and investigate the scaling of catalytic advantage with respect to catalyst size, we apply Thm.~\ref{thm:nice_order} to find the limits of the cooling performance for a qutrit when using catalysts of varying dimensions, ranging from two to thirty. The catalyst Hamiltonian is trivial for simplicity. 

\begin{figure}[t]
	\centering
	\includegraphics[width = 0.7\columnwidth]{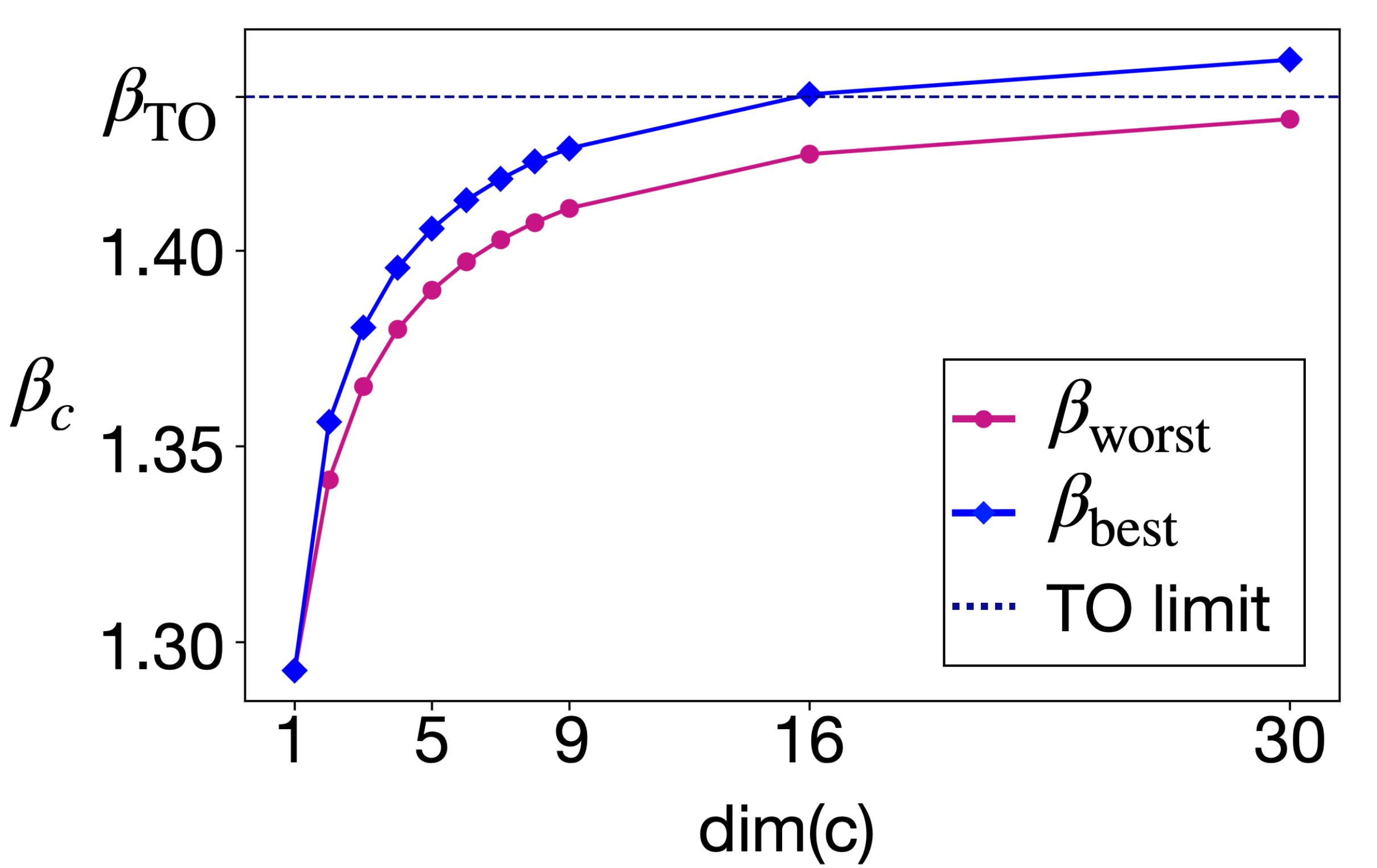}
	\caption{The cooling performance of CETO from a thermal state to another thermal state, quantified by final inverse temperature $ \beta_{c} $ attainable when the catalyst state is fixed. We set the initial system inverse temperature to be $ \beta_{h} = 0.5 $, and the ambient temperature to be $ \beta = 1 $. The system of interest is three-dimensional with energy levels $ (0,0.4,0.5) $ and catalysts of dimension $ {\rm dim}(\catstate) $ from two to nine, sixteen, and thirty are used. $ {\rm dim}(\catstate) = 1 $ case corresponds to non-catalytic ETO, and the blue dashed line above marks the inverse temperature $ \beta_{\rm TO} $ that non-catalytic TO can achieve. We searched over catalyst state distributions among minimally-disturbing catalysts. The results from the worst performing catalysts in each dimension are marked with purple circles, while blue diamonds are from the best catalysts in each dimension.}
	\label{fig:thermal_to_thermal}
\end{figure}

Fig.~\ref{fig:thermal_to_thermal} shows the coldest achievable $ \beta_{c} $ from minimally-disturbing catalysts, where the worst and the best cases are marked with purple circles and blue diamonds, respectively.
Even with qubit catalysts, almost half of the gap between the TO limit (dashed line) and the ETO limit ($ {\rm dim}(\catstate) = 1 $) is covered. The maximal catalytic advantage (blue solid line) gradually increases with the catalyst size, and at $ {\rm dim}(\catstate) = 16 $, best catalysts among the sample surpass TO limit, whilst at $ {\rm dim}(\catstate) = 30 $, most of the samples perform better than TO. Note that we have limited the range of catalyst distribution to fix the initial composite state $ \beta $-order; hence there might exist (not minimally-disturbing) catalyst states that activate a better cooling process than the ones marked in Fig.~\ref{fig:thermal_to_thermal}. Also, even the worst case catalysts provide some advantage from the same reason. Usually, if the catalyst is almost pure or pure, catalytic advantage vanishes. Nevertheless, our results give an efficiently computable lower bound to the achievable amount of cooling when any $ d $-dimensional catalyst is allowed.

To estimate the true minimum and maximum cooling performance achievable by the given set of catalysts, we employed multiple strategies, including uniform random sampling and gradient-descent-like search. Interestingly, regardless of the catalyst dimension, the worst catalysts are given by maximally mixed states which are the Gibbs states for the trivial Hamiltonian. Using Gibbs states as a catalyst can provide no advantage in the framework of TO, since they form the free states. However, for a set of operations with innate Markovianity, such as ETO, Gibbs states can activate state transitions as a catalyst by providing additional non-Markovianity. This is reminiscent of the discussion in Sec.~\ref{subsec:tracking}, where the catalyst functions as a temporary storage during the evolution. See~\cite{Czartowski23_GibbsCat} for a similar setup and~\cite{Son23_hierarchy} for the ultimate power of Gibbs state catalysts. The best catalyst distributions we obtained are non-trivial and would be of interest for future study. 

Overall, the results in this section display that small catalysts do provide substantial advantage in the setting of ETO, where simple two-level swaps are sufficient to execute the procedure. Also, even when the catalysts are not fine-tuned, when they are in a certain regime, catalysis still enhances ETO considerably. In particular, Gibbs states are assumed to be easily preparable, and thus can be good candidate states for a more realistic catalytic protocol.

\section{Characterizing $\clSeto$}\label{sec:characterising_Seto}

Currently, the only known way to determine whether $\pstate \xrightarrow{\rm ETO} \qstate$, is to construct the full set $\clSeto(\pstate)$ by finding all the extreme points of this convex polytope, and check if $\qstate\in\clSeto(\pstate)$. Lostaglio et al.~\cite{Lostaglio_18_ETO} provided a systematic way of finding all extreme states of $\clSeto(\pstate)$ for an arbitrary dimension $ d$, which involves an exhaustive search among all possible $\beta$-swap sequences with a bounded length $\ell_{\rm max}$. In other words, this procedure identifies an upper bound on the number of extreme points 
\begin{equation}\label{eq:number_of_ext}
	\# {\rm ~extremal~points}~ \leq {d \choose 2}^{\ell_{\rm max}}.
\end{equation}
In~\cite{Lostaglio_18_ETO}, $ \ell_{\rm max}\leq d! $ is shown, which means that Eq.~\eqref{eq:number_of_ext} grows super-exponentially with the dimension of the system.
This presents a serious roadblock to both understanding and determining the possibility of state transitions via ETO, a reason why ETO, despite its strong physical motivation, has not been extensively studied.

We improved this result in the following ways:
\begin{enumerate}
\item For the simplest non-trivial case of $d=3$, we provide a full characterisation of $\clSeto(\pstate)$, showing in particular that $\clSeto$ has at most 8 extreme points, in contrast to the upper bound in~\cite{Lostaglio_18_ETO} of $3^6 = 729$. 
\item For $d>3$, 
\begin{enumerate}
	\item we tighten the upper bound on $\ell_{\rm{max}}$ analytically by the factor of $d-3$ (see Thm.~\ref{thm:ETO_cone}). 
	\item for a subset of initial states where the $ \beta $-ordering is monotonic in energy, we derive a tight upper bound for the number of extreme points, and fully identify the corresponding $ \beta $-swaps (Thm.~\ref{thm:nice_order}). 
\end{enumerate}
\item We use computational algorithms to obtain heuristics by random sampling of initial states. We present a comparison of these results with theoretical analysis in Table \ref{table:len_upper_bound}.
\end{enumerate} 

\subsection{An exact characterisation for $ d=3 $}\label{sec:3d_ETO}
A qutrit system is the simplest setup where $\clSeto$ deviates from $\clSto$. Furthermore, the no-go results established for qutrits, which rule out certain swap series from generating extreme points of $\clSeto$, hold true for any three levels of higher-dimensional scenarios. Here, we provide a full characterisation of $\clSeto(\pstate)$ for any qutrit state $\pstate$ by defining two simple sets.

\begin{enumerate}
		\item The set containing the initial $ \pstate $ after undergoing not more than 2 non-identical neighbouring swaps:
		\begin{align}\label{eq:Theta}
			\Theta_\mathrm{ETO}^{(3)}(\pstate) = \{\pstate, \beta^{( \pi_1, \pi_2)}\pstate, \beta^{( \pi_2, \pi_3)}\pstate,
			\beta^{( \pi_1, \pi_3)}\beta^{( \pi_1, \pi_2)}\pstate,  
			\beta^{( \pi_1, \pi_3)}\beta^{( \pi_2, \pi_3)}\pstate\}.
		\end{align}
		
		\item The set of states that undergo three neighbouring swaps or a non-neighbouring swap:
		\begin{align}\label{eq:Xi}
			\Xi_\mathrm{ETO}^{(3)}(\pstate) = \{\beta^{( \pi_1, \pi_3)}\pstate, \beta^{( \pi_2, \pi_3)}\beta^{( \pi_1, \pi_3)}\beta^{( \pi_1, \pi_2)}\pstate,
			\beta^{( \pi_1, \pi_2)}\beta^{( \pi_1, \pi_3)}\beta^{( \pi_2, \pi_3)}\pstate\}.
		\end{align}
\end{enumerate}

\begin{theorem}\label{thm:3dETO}
		For any $\pstate\in\probspace^3$ with $\border(\pstate)_k = \pi_k$, let $ \Theta_\mathrm{ETO}^{(3)} $ and $ \Xi_\mathrm{ETO}^{(3)} $ be defined as in Eqs.~\eqref{eq:Theta} and~\eqref{eq:Xi}, where we have dropped the explicit $ \pstate $ for notational convenience. 
		Then,  
		\begin{equation}\label{key}
			\Theta_\mathrm{ETO}^{(3)} ~\subset~ {\rm Ext}(\clSeto) ~\subset~ \Theta_\mathrm{ETO}^{(3)} \cup \Xi_\mathrm{ETO}^{(3)}.
		\end{equation} 
\end{theorem} 

One of the important techniques used in proving Thm.~\ref{thm:3dETO} is identifying the $ \beta $-swap series that produce states which are extreme points of $ \clSto $. As a result of Thm.~\ref{thm:TO_cone} and Cor.~\ref{corollary:TO-ETO_ext}, such states are \emph{unique} extreme points for final states of their particular $\beta$-ordering. This identification utilises results established in the language of biplanar transportation matrices in~\cite{Mazurek19_channels}. We leave the full proof for \ref{appendix:3dETO_lemmas} and focus on a few notable points made there. 
Firstly, observe that even non-neighbouring swaps can generate extreme points. 
This fact rules out the naive attempt of working only with states of the form $\qstate^{(\vec{j})}$, and forces us to consider the entire set of swap series, which is in contrast to the case of MTP, where the extreme state verifying algorithm only needs to search for neighbouring swaps~\cite{Lostaglio22_MTP1}.
However, a necessary condition for $\beta^{( \pi_1, \pi_3)}\pstate$ to be extremal in $\clSeto(\pstate)$ exists and alleviates the complications. 
\begin{lemma}\label{lemma:order_swap_lemma}
	For any $\pstate\in\probspace^d$ with $\border(\pstate)_k = \pi_k$, the state $\beta^{( \pi_i, \pi_{i+c})}\pstate$ can be an extreme point of $\clSeto(\pstate)$ only if $\border(\beta^{( \pi_i, \pi_{i+c})}\pstate) = S_{i,i+c}(\border(\pstate))$, where $ S_{i,j} $ is the operation that swaps the $ i $th element with the $ j $th element.
\end{lemma}
This lemma can be proven by straightforward calculations of swapped states, as described in \ref{appendix:order_swap_lemma}. Two simplifications emerge from this result. Firstly, in numerical computations, all non-neighbouring swaps that give orders different from what the lemma imposes can be removed. Secondly, when analytical studies are carried out without specifying an initial state, the ambiguity of $\beta$-orders after the swap vanishes when the resulting state is shown to be an extreme point of $\clSeto$.

\subsection{Improved characterisation for higher-dimensions}
Finding the extreme points of $ \clSeto(\pstate) $ is generally a hard task. In particular, multiple extreme points that correspond to the same $ \beta $-ordering often exist, and the total number of extreme points is only known to be upper bounded by Eq.~\eqref{eq:number_of_ext}, which yields an extremely loose number, scaling as $d^{2d!/(d-3)}$ even with our improved result Thm.~\ref{thm:ETO_cone}. Lemma~\ref{lemma:ext_point_for_each_order} provides some intuition into the lower bound on the number of extreme points for generic initial states.
\begin{lemma}\label{lemma:ext_point_for_each_order}
	For any state $ \pstate\in\probspace^d $, the reachable state set $ \clSeto(\pstate) $ has at least one extreme point $ \rstate $ having $ \border(\rstate) = \psi $ for any ordering $ \psi $. 
\end{lemma}
The proof can be found in \ref{appendix:proof_ext_point_for_each_order}. 
Note that Lemma~\ref{lemma:ext_point_for_each_order} does not imply that all initial states have at least $ d! $ distinct extreme points. For instance, a Gibbs state $ \tbeta $ cannot be transformed into any other state via ETO and thus $ \clSeto(\tbeta) $ has a single (extreme) point $ \tbeta $. In this case, $ \slope(\tbeta) $ is fully degenerate and $ \border(\tbeta) $ can be any permutation of $ (1,\cdots,d) $.

Similar to the $ \tbeta $ example, some initial states admit a simpler $ \clSeto $ structure.
One of our main analytical results shows that the situation simplifies drastically and achieves the lower bound of Lemma~\ref{lemma:ext_point_for_each_order}, when $ \pstate $ is known to have a particular $ \beta $-ordering, i.e. when it is \emph{monotonic in energy},
	\begin{equation}\label{eq:mono_order}
		\border(\pstate) = (1,\cdots,d) ~ \text{or} ~ \border(\pstate) = (d,\cdots,1). 
	\end{equation}
\begin{theorem}\label{thm:nice_order} 
	If $\border(\pstate) $ is monotonic in energy, extreme points of $\clSeto(\pstate)$ are achieved if and only if the corresponding $\beta$-swap series that produce them are
	\begin{enumerate}
		\item always neighbouring, 
		\item containing no repetition of each swap.
	\end{enumerate} Furthermore, when $ \vec\beta_1\pstate, \vec\beta_2\pstate \in \mathrm{Ext}[\clSeto(\pstate)] $ and $ \border(\vec\beta_1\pstate) = \border(\vec\beta_2\pstate)$ for such $ \pstate $, the two series are identical ($ \vec\beta_1 = \vec\beta_2 $). 
\end{theorem}
See \ref{app:special_ordering} for the proof.
Several important simplifications follow from the above lemma, whenever $\pi(\pstate)$ is monotonic in energy. From the no-repetition condition, $\ell_{\rm max} = d(d-1)/2$ is obtained. The equivalence of $ \beta $-swaps outputting the same target state $ \beta $-order also guarantees the uniqueness of extreme points of $ \clSeto(\pstate) $ at each order, setting the maximum number of extreme points to be $ d! $. 
More importantly, given the target $ \beta $-order, one can immediately identify a corresponding extreme point without the need of searching over all possible series, since we developed an explicit algorithm to evaluate this extreme point, which we call the standard formation (see Def.~\ref{def:standard_formation} for details).

Notably, all thermal states $ \tau^{\beta^{\prime}} $, with temperature $ 1/\beta^{\prime}$ that might be different from the ambient temperature $ 1/\beta $, satisfy the monotonicity of the $ \beta $-ordering. 
The analysis in Secton~\ref{sec:higherd_cat} deals with at most ninety dimensional system-catalyst composites. Without Thm.~\ref{thm:nice_order}, evaluating the extreme points for such high dimensional systems is practically impossible. However, leveraging the nice ordering of the initial state, we were able to simulate over two million ninety dimensional systems in just a few hours.

For the case of generic initial states, we derive an improved upper bound on the $ \beta $-swap series length that generates the ETO cone. 
\begin{theorem}[Improvement of Thm.~6 of~\cite{Lostaglio_18_ETO}]\label{thm:ETO_cone} 
	$\clSeto(\pstate)$ is the convex hull of all final states generated by 
	\begin{align}
		\left\{\beta^{(j_l,k_l)}\cdots \beta^{(j_1,k_1)}\pstate\right\}_{l=0}^{l_\mathrm{max}},\label{eq:full_set_bseries}
	\end{align}
	with $l_\mathrm{max} \leq \left\lfloor\frac{d!-4}{d-3}\right\rfloor$ for $d\geq 4$ and all possible combinations of $\{(j_n,k_n)\}$. 
\end{theorem}
The proof of Thm.~\ref{thm:ETO_cone} can be found in \ref{appendix:ETO_cone_proof}. The main technique is again to identify a set of states that thermomajorise all the other states in $ \clSeto $ that share the same $ \beta $-order, either via applying the results of \cite{Mazurek19_channels}, or showing that they are tightly-majorised by the previous state.

However, even for generic initial states, most elements in Eq.~\eqref{eq:full_set_bseries} are not extremal in $\clSeto(\pstate)$. As discussed in Sec.~\ref{subsec:neighbouring}, non-neighbouring swaps often deplete athermality too rapidly and the resulting state ends up being close to equilibrium and far from extremality.

\begin{table*}[t]
	\centering
		\begin{tabular}{|c||c|c|c|c|}
			\hline
			$d =$ dim($\pstate$) &$\ell_{\max}$ (\cite{Lostaglio_18_ETO})  &  $\ell_{\max}$ (Thm.~\ref{thm:ETO_cone})  & $\ell_{\max}$ (heuristics) & $N_{\mathrm{Ext}}$ (heuristics) \\
			\hline
			3& 6 &  3 (Thm.~\ref{thm:3dETO}) & 3  & $\mathcal{O}(10^0)$\\
			\hline
			4 & 24 & 20& 8 & $\mathcal{O}(10^1)$\\
			\hline
			5 & 120 & 58 & 16  & $\mathcal{O}(10^2)$\\
			\hline
			6 & 720 & 238 & 23 & $\mathcal{O}(10^3)$\\
			\hline 
			7 & 5040 & 1259 & 38 & $\mathcal{O}(10^5)$\\
			\hline
		\end{tabular}
		\caption{Theoretical bounds and numerical results for maximum $\beta$-swap lengths $l_{\max}$ needed to construct all extreme points of $\clSeto(\pstate)$ are compared up to $d = 7$. Unlike the factorial scaling of the theoretical bound, heuristics suggest polynomial growth of the length. $N_{\mathrm{Ext}}$ is the typical number of extreme points of $\clSeto(\pstate)$, also obtained from numerical search.}
		\label{table:len_upper_bound}
	\end{table*}

\subsection{Gap between current theoretical characterisation and heuristical analysis}\label{subsec:numerics_high_d}

When naively constructing all $\beta$-swaps as given in Thm.~\ref{thm:ETO_cone}, the number of the series that need to be checked scales as $(d/\sqrt{2})^{2d!/(d-3)}$, which is practically infeasible after $d=4$. However, a simplification can be made by exploiting the fact that if a series $\vec{\beta}$ yields a non-extreme point $\vec{\beta}\pstate = [\lambda_1\vec{\Mswap}_1 + (1-\lambda_1)\vec{\Mswap}_2]\pstate$, all states $\vec{\beta}^\prime\vec{\beta}\pstate$ with any $\beta$-swap series $\vec{\beta}^\prime$ are also non-extreme. Thus by starting from length-1 series and checking extremality after each time the series is lengthened, we can greatly reduce the computational cost.

There are two ways to verify the extremality of final states. For $d=3,4$ and qutrit-qubit cases, we explicitly used the fact that the resulting set is convex:
\begin{enumerate}
	\item Apply length-1 $\beta$-swaps to the set $\clS_0 = \{\pstate\}$ to get $\clS_1 =\{\beta^{(k,l)}\pstate\}_{\forall k,l}$. Construct the hull $\mathrm{Conv}[\clS]$, where $\clS = \bigcup_{j}\clS_j$.
	\item Find extreme points $\mathrm{Ext[Conv[\clS]]}$ and update $\clS_1 \rightarrow \clS_1 = \mathrm{Ext[Conv[\clS]]} - \clS_0$. 
	\item Iterate steps 1 and 2 by applying single swaps to $\clS_i$, save newly obtained extreme states as $\clS_{i+1}$, and update $\clS_{j\leq i+1}$ by eliminating non-extreme points. Update $\clS = \bigcup_{j}\clS_j$.
	\item Stop when $\clS_{i+1} = \emptyset$. 
\end{enumerate}
The strength of this algorithm is its straightforwardness. It does not include any non-extreme point nor does it miss an extreme point of $\clSeto(\pstate)$. However, especially in the higher-dimensional cases, constructing convex hulls and characterising extreme points are highly demanding. 

The second algorithm, which we used for $d = 5,6,7$, is a slight modification of the one presented in~\cite{Lostaglio22_MTP1} for MTP and utilises Lemma~\ref{lemma:same_border}:
\begin{enumerate}
	\item Apply length-1 $\beta$-swaps to the set $\clS_0 = \{\pstate\}$ to get $\clS_1 =\{\beta^{(k,l)}\pstate\}_{\forall k,l}$. 
	\item Apply length-1 $\beta$-swaps to the set $\clS_1$ to obtain $\clS_2$ and check thermomajorisation relations between states in $\clS = \bigcup_{j}\clS_j$ that share the same $\beta$-order. Eliminate the ones that are thermomajorised by other states. 
	\item Iterate step 2 for $i$ times until $\clS_{i+1} = \emptyset$. Update $\clS = \bigcup_{j}\clS_j$. 
\end{enumerate}
The difference from the algorithm for MTP is that our algorithm also considers non-neighbouring swaps. The second strategy has an important edge over the previous one: there are only $m(d-1)$ inequalities to check for each final state, where $m$ is the number of already obtained extreme points having the same order. Hence for higher-dimensional systems, the second approach is preferable. Yet, non-extreme states that are achievable by convex combinations of different extreme states, but not thermomajorised by either of them, can be included in the final set. When interested only in the final set $ \clSeto $, counting some non-extreme states is still permissible.

To search for the maximum length $\ell_{\max}$, randomly generated initial states and energy levels were used to construct the set $\clSeto$. For $d\leq 6$, $\sim\mathcal{O}(10^4)$ random samples were tested for each dimension. For $d = 7$, only $200$ cases were calculated due to its high computational cost ($\sim$ few hours for each initial state). Table~\ref{table:len_upper_bound} presents results from the search, where $\ell_{\rm max}$ scales much slower than the theoretical prediction. Moreover, the typical number of extreme points also deviates significantly from the naive expectation, i.e. they do not scale exponentially with $\ell_{\rm{max}}$.

\section{Discussions and Conclusions}\label{sec:discussion}




In search of more practical thermal processes, we analyse catalysis in elementary thermal operations with small catalysts. These operations offer a clearer path to implementation and less stringent experimental control requirements, compared to the more general and well-studied framework of thermal operations. However, it is known that ETO is only a strict subset of TO, and in particular, some of the states reachable via TO are no longer achievable via ETO. We try to alleviate this limitation by allowing catalysis while maintaining ease of execution by limiting the catalyst size.

Several roadblocks had to be removed in order to tackle this problem. Firstly, there is currently no efficient way of characterising the set of allowed state transitions under ETO. We partially overcome this challenge through two approaches. First, we fully solve the characterisation problem for three-dimensional systems and for a subset of initial states of generic dimension. Additionally, we improve the analytical upper bound of the computational cost needed for the most general cases. Armed with these tools, we demonstrate the existence of catalysis in ETO, where even small catalysts prove to be remarkably useful. In fact, the extreme case of employing a qubit catalyst alone nearly eliminates the gap between TO and ETO for qutrit system states, sometimes even providing additional advantages beyond TO. At the same time, our physically relevant example of a cooling protocol highlights the power of relatively small catalysts, approaching the TO limit without the need for fine-tuning the catalyst states. 

Another obstacle is the initial lack of clarity on why catalysts work, despite reports of their existence (or non-existence) in various resource theories. To address this, we leverage the step-wise structure of ETO to track and analyse catalytic evolutions, by capturing snapshots of states after each ETO step. This approach opens up a new avenue for understanding the underlying origins of catalytic advantage. In our example, the catalyst's role was to receive the free energy flowing out from the system, either through reduced state population changes or correlations with the system. Without the catalyst, all changes in system free energy would dissipate into the surrounding bath, which thermalises after each swap. This interpretation could potentially be extended to catalysts in different resource theories. For instance, it would be intriguing to further investigate how residual correlations between catalyst and system (for the scenario of correlating catalysis) are exemplified in this picture.

Among our results, the theoretical upper bound presented in Thm.~\ref{thm:ETO_cone} is expected to have more room for tightening. The $\beta$-swapping operations (unless the energy levels are degenerate) are highly resource-depleting in general. Therefore, the application of these operations exponentially many times, e.g. using $\ell_{\max}$ number of times, is unlikely to produce a final state that is extreme in terms of athermality. Our heuristic approach strongly suggests that the maximum length scales polynomially, rather than exponentially. Yet, we currently do not have rigorous proof of this scaling, and we leave it for future studies. 



\section{Acknowledgements}
We thank Matteo Lostaglio for insightful discussions. This work was supported by the start-up grant of the Nanyang Assistant Professorship of Nanyang Technological University, Singapore.

\newcommand{\newblock}{}

\newpage
\appendix

We organise the appendix as follows: \ref{appendix:dynamics} illustrates the real time dynamics of one elementary thermal operation step, assuming intensity-dependent Jaynes-Cummings interaction between two levels and a harmonic oscillator in Gibbs state.

\ref{appendix:ddim_ETO} contains various technical results we developed on elementary thermal operations (ETO), which are necessary for the analytical results in the main text. 
In particular, the proof of Lemmas~\ref{lemma:order_swap_lemma} is provided in this section.

In \ref{appendix:3dETO_lemmas}, we develop the full characterisation of ETO with $ d=3 $ by proving Thm.~\ref{thm:3dETO}. For general high-dimensional cases, we prove Thm.~\ref{thm:ETO_cone} in \ref{appendix:ETO_cone_proof}. 

\ref{app:special_ordering} proves Thm.~\ref{thm:nice_order} claiming the drastic simplification of the $ \clSeto $ extreme points construction for two special initial $ \beta $-orders that are monotonic with the energy eigenvalues.

In \ref{appendix:CETO}, we detail the methods of constructing the set of reachable states via catalytic ETO with a qubit catalyst. Here we also analyse an example transition obtained from this method. 



\section{Dynamics of an elementary thermal operation}\label{appendix:dynamics}
In this section, we show that any two-level swap $ \Mswap_{\lambda}^{(i,j)} $ can be implemented through the intensity-dependent Jaynes-Cummings interaction between two levels to be swapped and a harmonic oscillator with the matching energy gap. 
Our model comprises two-level system, harmonic oscillator, and an interaction between them, which respectively read
\begin{align}
	H_{S} &= E_1 \ketbra{1}{1}_{S} + (E_1+\omega)\ketbra{2}{2}_{S},\\
	H_{B} &= \sum_{n=0}^{\infty}n\omega\ketbra{n}{n}_{B},\\
	H_{SB} &= \sum_{n=1}^{\infty}g\ketbra{1}{2}_{S}\otimes \ketbra{n}{n-1}_{B} + \text{h.c.} \equiv \sum_{n=1}^{\infty}gX^{(n)},
\end{align}
where 
\begin{equation}\label{key}
	X^{(n)}\equiv \ketbra{1}{2}_{S} \otimes \ketbra{n}{n-1}_{B} + \text{h.c.}
\end{equation} 
is a Pauli $ X $-like operator in the subspace spanned by energy eigenvectors $ \ket{1}_{S}\ket{n}_{B} $ and $ \ket{2}_{S}\ket{n-1}_{B} $ possessing the same energy and $ g $ is the coupling strength parameter. Note that the interaction term is energy-preserving, i.e.
\begin{equation}\label{key}
	 [H_{S}\otimes\mathbb{1}_{B}+\mathbb{1}_{S} \otimes H_{B}, H_{SB}] = 0 .
\end{equation} 
For simplicity, we shift to the interaction picture. We always consider the initial state to be a product of incoherent system state and a Gibbs state w.r.t. some inverse temperature $ \beta $,
\begin{equation}\label{key}
	\varrho_{SB} = \left(p_{1}\ketbra{1}{1}_{S}+p_{2}\ketbra{2}{2}_{S}\right)\otimes \sum_{n=0}^{\infty}\left(1-e^{-\beta\omega}\right)e^{-n\beta\omega}\ketbra{n}{n}_{B},
\end{equation}
which also commutes with $ H_{S}+H_{B} $. Then, the state $ \varrho_{SB} $ and the Hamiltonian $ H_{SB} $ remain invariant in the interaction picture. The time-evolution operator is given by 
\begin{align}
	U(t) &= e^{-iH_{SB}t} = \prod_{n=1}^{\infty}e^{-igX^{(n)}t} = \prod_{n=1}^{\infty}\left[\mathbb{1}_{SB} + (\cos(gt)-1)\mathbb{1}^{(n)} - i\sin(gt)X^{(n)} \right]\nonumber\\
	&= \mathbb{1}_{SB} + \sum_{n=1}^{\infty}\left[(\cos(gt)-1)\mathbb{1}^{(n)} - i\sin(gt)X^{(n)} \right], 
\end{align}
making use of the properties of $ X^{(n)} $,
\begin{equation}\label{key}
	X^{(n)}X^{(m)} = 0\ \text{for}\ n\neq m,\quad \left(X^{(n)}\right)^2 = \ketbra{1}{1}_{S}\otimes\ketbra{n}{n}_{B} + \ketbra{2}{2}_{S}\otimes\ketbra{n-1}{n-1}_{B} \equiv \mathbb{1}^{(n)}.
\end{equation}

After time $ t $, the system reduced state becomes
\begin{align}
	\rho_{S}(t) &= \Tr_{B}\left[U(t)\varrho_{SB}U(t)^{\dagger}\right] \label{eq:conti_evolution1}\\
	&= \left(\left[1-e^{-\beta\omega}\sin^{2}(gt)\right]p_{1}+\sin^{2}(gt)p_{2}\right)\ketbra{1}{1}_{S} + \left(e^{-\beta\omega}\sin^{2}(gt)p_{1}+\cos^{2}(gt)p_{2}\right)\ketbra{2}{2}_{S}\nonumber\\
	&\equiv q_{1}(t)\ketbra{1}{1}_{S} + q_{2}(t)\ketbra{2}{2}_{S}. \nonumber
\end{align}
Since the reduced state on $ S $ is quasi-classical at all time, we focus on the population vector $ \qstate(t) = (q_{1}(t),q_{2}(t)) $, which can be written as a elementary thermal operation channel
\begin{equation}\label{eq:conti_evolution2}
	\qstate(t) = \Mswap^{(1,2)}_{\lambda(t)}\pstate,
\end{equation}
with $ \lambda(t) = \sin^{2}(gt) $. Eq.~\eqref{eq:conti_evolution2} specifies the continuous evolution during the two-level swap process, achieving the $ \beta $-swap at $ gt = \pi/2 $.  

%

\section{Elementary thermal operations for $d$-dimensional systems}\label{appendix:ddim_ETO}

\subsection{Basic transformation of thermomajorisation curves after a $\beta$-swap}\label{appendix:therm_curves_bswaps}

Given a system characterised by Hamiltonian $ H $, we denote its thermal state as $ \tbeta $, and describe the initial state with respect to its energy population vector $ \pstate $. Under a $ \beta $-swap $\beta^{(k,l)}$ as defined in Eq.~\eqref{eq:beta_jk}, 
\begin{align}\label{eq:betaswap}
	\begin{pmatrix}
		q_k\\q_l
	\end{pmatrix} &\equiv \beta^{(k,l)}\begin{pmatrix}
		p_{k}\\p_{l}
	\end{pmatrix}
	= \begin{pmatrix}
		(1-\Delta_{kl})p_{k}+p_{l} \\ \Delta_{kl}p_{k}
	\end{pmatrix},\ \text{and }	q_m = p_m\ \forall m\neq k,l, 
\end{align}
when $E_k\leq E_l$.
Then the element-wise ratios, $ \slope(\qstate) $ as defined in Eq.~\eqref{eq:slopes}, are transformed accordingly
\begin{align}
	\slope(\qstate)_{k} &= (1-\Delta_{kl})\slope(\pstate)_{k}+ \Delta_{kl}\slope(\pstate)_{l}\label{eq:slope_after_swap1},\\
	\slope(\qstate)_{l} &= \slope(\pstate)_{k},\label{eq:slope_after_swap2}\\
	\slope(\qstate)_{m} &= \slope(\pstate)_{m},\ \forall m\neq k,l\label{eq:slope_after_swap3},
\end{align}
where 
\begin{equation}
	\begin{cases}
		\slope(\pstate)_k\geq\slope(\qstate)_k\geq\slope(\pstate)_l, \quad \text{if}\quad \slope(\pstate)_k\geq\slope(\pstate)_l,\\
		\slope(\pstate)_k\leq\slope(\qstate)_k\leq\slope(\pstate)_l,\quad \text{if}\quad \slope(\pstate)_k\leq\slope(\pstate)_l.
	\end{cases}\label{eq:slope_after_swap4}
\end{equation}
Furthermore, equalities for the above equations hold only under the following circumstances:
	\begin{align}
		\slope(\pstate)_k=\slope(\pstate)_l\qquad &\implies \qquad 	\slope(\pstate)_k = \slope(\qstate)_k = \slope(\pstate)_l, \\
		E_k=E_l \qquad&\implies\qquad  \slope(\qstate)_k = \slope(\pstate)_l.
\end{align}  
Naturally, the $ \beta $-swap operations also alter $ \beta $-orderings of states.
Let us denote the initial $\beta$-order as $\border(\pstate) = (\pi_1,\cdots,\pi_d)$. 
If $k = \pi_i$ and $l = \pi_{i \pm 1}$ for some $i$, that is if $\beta^{(k,l)}$ is a neighbouring swap for a state $\pstate$, then $\border(\qstate)$ can be easily determined: 
	\begin{align}\label{eq:order_after_swap1}
		\border(\qstate)= S_{i,i \pm 1} (\border(\pstate)),
	\end{align}
where $ S_{i,j} $ is a swap between $ i $'th and $ j $'th elements as introduced in Lemma~\ref{lemma:order_swap_lemma}.

\subsection{Useful technical remarks}\label{appendix:sec2_misc}

Here we present sundry remarks on $\beta$-swaps that are utilised in proofs of lemmas and theorems in the later part of the appendix. These results hold as equalities in the channel level and do not depend on the states these channels are acting on.

\begin{remark}\label{rmk:no_repetition}
	The $ \beta $-swap series $\beta^{(k,l)}\beta^{(k,l)}=\Mswap_{\lambda}^{(k,l)}$ for some $\lambda\neq 0,1$, and thus always produces a non-extreme point except for in the trivial case where $E_k = E_l$. In that trivial case, two repeated swaps always result in identity, and $ \lambda=0 $.
\end{remark}

\begin{remark}\label{rmk:degenerate_swap}
	If $E_k=E_l$ for some levels $k,l$, $\beta^{(k,l)}$ always connect one extreme point of $\clSeto$ to another.
\end{remark}

\begin{proof}
	If this statement does not hold, there exists a state $\qstate$ that is extremal for $\clSeto(\pstate)$ but $\beta^{(k,l)}\qstate = \sum_i\lambda_i \rstate_i$ for some $\lambda_i\geq 0$ and $\rstate_i\in\clSeto(\pstate)$. But then $\qstate = (\beta^{(k,l)})^2\qstate = \sum_i\lambda_i\beta^{(k,l)}\rstate_i$. Since $\beta^{(k,l)}\rstate_i\in\clSeto(\pstate)$, $\qstate$ cannot be extremal.  
\end{proof}

\begin{remark}\label{rmk:commute}
	When $k,l,m,n$ are all distinct, $ \beta $-swaps commute, i.e.
	\begin{equation}
		\beta^{(k,l)}\beta^{(m,n)} = \beta^{(m,n)}\beta^{(k,l)}.
	\end{equation}
\end{remark}

\begin{remark}\label{rmk:swap_series_eq}
	If $E_{k}\leq E_{l}\leq E_{m}$, 
	\begin{equation}
		\beta^{(k,l)}\beta^{(k,m)}\beta^{(l,m)} = \beta^{(l,m)}\beta^{(k,m)}\beta^{(k,l)}.
	\end{equation}
\end{remark}
The equality is obtained through direct calculations,
\begin{align}
	\beta^{(k,l)}\beta^{(k,m)}\beta^{(l,m)} = \beta^{(l,m)}\beta^{((k,m)}\beta^{(k,l)}
	= \begin{pmatrix}
		(1-\Delta_{kl})(1-\Delta_{km}) & 1-\Delta_{km} & 1\\
		\Delta_{kl}(1-\Delta_{km}) & \Delta_{km} & 0\\
		\Delta_{km} & 0 & 0
	\end{pmatrix},
\end{align}
where we omit identities acting on irrelevant levels when writing ETO maps (and do so consistently in the rest of the appendix for notational brevity).

\subsection{Proof of Lemma \ref{lemma:neighbouring_TOext} }\label{app:proof_lem3}
\begin{proof}
For the first part of the Lemma: by  Eq.~\eqref{eq:order_after_swap1}, a single neighbouring swap $ \beta^{\pi_{j},\pi_{j+1}} $gives us 
\begin{align}
	\sum_{i=1}^{m}\tbeta_{\pi_i} = \sum_{i=1}^{m}\tbeta_{\border(\qstate)_i}, \quad \forall m\neq j,
\end{align}
which, combined with Eqs.~\eqref{eq:slope_after_swap1}--\eqref{eq:slope_after_swap3}, yields 
\begin{equation}	
	\mathcal{L}_\qstate(a_m) = \mathcal{L}_\pstate(a_m),\quad a_m\equiv\sum_{i=1}^{m}\tbeta_{\border(\qstate)_i},\quad m = 1,2,\cdots,d.
\end{equation}
Each $ (a_m,\mathcal{L}_\qstate(a_{m})) $ corresponds to an elbow point of $ \mathcal{L}_\qstate $, which coincides with $ \mathcal{L}_\pstate $, implying that $\qstate$ is tightly thermomajorised by $\pstate$.

For the second case when $\beta^{(\pi_{i},\pi_{i+c})}$ acts on non-neighbouring levels ($ c>1 $) in $\pstate$, the change in $ \border(\qstate) $ becomes a little less straightforward. That is, 
\begin{equation}
	\border(\qstate) \neq S_{i,i+c}(\border(\pstate)),
\end{equation}
in general.
First suppose $ E_{\pi_{i}} < E_{\pi_{i+c}}$. Then 
\begin{equation}
	\slope(\qstate)_{\pi_{i+c}} = \slope(\pstate)_{\pi_{i}}> \slope(\qstate)_{\pi_{i}} = (1-\Delta_{\pi_{i}\pi_{i+c}})\slope(\pstate)_{\pi_{i}} + \Delta_{\pi_{i}\pi_{i+c}}\slope(\pstate)_{\pi_{i+c}} > \slope(\pstate)_{\pi_{i+c}},\label{eq:slope_after_nnswap}
\end{equation}
which implies $ \pi_{i+c} = \border(\qstate)_{i} $ and $ \pi_{i} = \border(\qstate)_{i+c^{\prime}} $ with $ 1\leq c^{\prime}\leq c $.
Here we imposed  
\begin{equation}\label{eq:strict_concavity}
	\slope(\pstate)_{\pi_{i}}>\slope(\pstate)_{\pi_{i+1}}\geq \slope(\pstate)_{\pi_{i+c-1}}>\slope(\pstate)_{\pi_{i+c}},
\end{equation}  
to make $ \pi_{i} $ and $ \pi_{i+c} $ truly non-neighbouring.
From the first equality of Eq.~\eqref{eq:slope_after_nnswap}, $ \mathcal{L}_{\qstate}(a_{i}) = \mathcal{L}_{\pstate}(a_{i}) $ follows. However, $ (i+1) $'th elbow for $ \mathcal{L}_\qstate $ is strictly separated from $ \mathcal{L}_\pstate $, which is sufficient to prove the second part of the Lemma. We demonstrate this for each possible case:\\

\emph{Case i:} $a_{i+1} - a_{i-1} \leq \tbeta_{\pi_{i}}$. Since $\slope(\pstate)_{\pi_{i}}>\slope(\qstate)_{\border(\qstate)_{i+1}}$,
\begin{equation}
	\mathcal{L}_\pstate(a_{i+1}) = \mathcal{L}_\pstate(a_i)+\slope(\pstate)_k \tbeta_{\border(\qstate)_{i+1}}>  \mathcal{L}_\pstate(a_i)+\slope(\qstate)_{\border(\qstate)_{i+1}} \tbeta_{\border(\qstate)_{i+1}} = \mathcal{L}_\qstate(a_{i+1}).
\end{equation}

\emph{Case ii-a):} $a_{i+1} - a_{i-1} > \tbeta_{\pi_{i}}$ and $c^{\prime}\neq1$. Then $ \border(\qstate)_{i+1} = \pi_{i+1}$. Using $\slope(\pstate)_{\pi_{i}} > \slope(\pstate)_{\pi_{i+1}}$,
\begin{align}
	\mathcal{L}_\pstate(a_{i+1}) &= \mathcal{L}_\pstate(a_i)+\slope(\pstate)_{\pi_{i}} (\tbeta_{\pi_{i}}-\tbeta_{\pi_{i+c}}) + \slope(\pstate)_{\pi_{i+1}} (\tbeta_{\pi_{i+1}} - \tbeta_{\pi_{i}} + \tbeta_{\pi_{i+c}})\\
	&= \mathcal{L}_\pstate(a_i)+(\slope(\pstate)_{\pi_{i}} -\slope(\pstate)_{\pi_{i+1}}) (\tbeta_{\pi_{i}}-\tbeta_{\pi_{i+c}}) + \slope(\pstate)_{\pi_{i+1}} \tbeta_{\pi_{i+1}}\nonumber\\
	&>  \mathcal{L}_\pstate(a_i)+\slope(\pstate)_{\pi_{i+1}} \tbeta_{\pi_{i+1}} = \mathcal{L}_\qstate(a_{i+1}).\nonumber
\end{align}

\emph{Case ii-b):} $ \border(\qstate)_{i+1} = \pi_{i} $, i.e. $a_{i+1} = a_i + \tbeta_{\pi_{i}}$. From $p_{\pi_{i}}+p_{\pi_{i+c}} = q_{\pi_{i}}+q_{\pi_{i+c}}$,
\begin{align}
	\mathcal{L}_\qstate(a_{i+1}) = \mathcal{L}_\pstate(a_{i-1}) + p_l + p_k = \mathcal{L}_\pstate(a_{i-1}+\tbeta_{\pi_{i}}) +
	\slope(\pstate)_{\pi_{i+c}}\tbeta_{\pi_{i+c}}. 
\end{align}
\hspace{2.2cm} The strict concavity of $\mathcal{L}_\pstate$ (Eq.~\eqref{eq:strict_concavity}) imposes 
\begin{align}
	\frac{\mathcal{L}_\pstate(a_{i+1}) - \mathcal{L}_\pstate(a_{i-1}+\tbeta_{\pi_{i}})}{\tbeta_{\pi_{i+c}}} > \frac{\mathcal{L}_\pstate(a_{i+c}) - \mathcal{L}_\pstate(a_{i+c} - \tbeta_{\pi_{i+c}})}{\tbeta_{\pi_{i+c}}} = \slope(\pstate)_{\pi_{i+c}},
\end{align}
\hspace{2.2cm} which leads to $\mathcal{L}_\pstate(a_{i+1})> \mathcal{L}_\qstate(a_{i+1})$.\\

The argument above is easily generalizable for $E_{\pi_{i}}>E_{\pi_{i+c}}$, so this concludes the proof.
\end{proof}

\subsection{Proof of Lemma~\ref{lemma:order_swap_lemma}}\label{appendix:order_swap_lemma}
Here, we want to prove that a $ \beta $-swap involving two levels $ E_j,E_k $ produces an extreme point only if its sole effect on the $ \beta $-ordering on the final state is a swap of $ j $ and $ k $. This technical result is later used in establishing Lemma \ref{rmk:only_one_nn}, a key tool used throughout in several subsequent proofs. We start by proving the lemma for the case of $\pstate\in\probspace^3$.
Denote the initial $\beta$-order as $\border(\pstate) = (\pi_1,\pi_2,\pi_3)$. Consider the following three cases:
\begin{enumerate}
	\item Suppose $E_{\pi_1} < E_{\pi_3}$ and $\border(\qstate) = (\pi_3,\pi_1,\pi_2)$, where $\qstate = \beta^{(\pi_1,\pi_3)}\pstate$. The same final ordering is obtained after two neighbouring swaps, i.e. 
	\begin{equation}\label{key}
		\qstate'= \beta^{(\pi_1,\pi_3)}\beta^{(\pi_2,\pi_3)}\pstate, \qquad \pi(\qstate') =  \pi(\qstate).
	\end{equation}
	Note that $q_{\pi_3} = q'_{\pi_3}$~(this is seen from Eq.~\eqref{eq:betaswap} and \eqref{eq:Deltas}). Therefore, the thermomajorisation curves $  \mathcal{L}_\qstate,\mathcal{L}_{\qstate'}$ are identical up to the first elbow
	\begin{equation}
		\mathcal{L}_\qstate(x) = \mathcal{L}_{\qstate'}(x), \qquad x\in[0,\tbeta_{\pi_3}].
	\end{equation}
	The third elbow also coincides as $\mathcal{L}_\qstate(1) = \mathcal{L}_{\qstate'}(1) = 1$. Finally, the second elbow points of $\qstate$ and $\qstate'$ curves are given by 
	\begin{equation}\label{eq:lemma6_proof_eq1}
		\mathcal{L}_\qstate(1-\tbeta_{\pi_2}) = 1 - p_{\pi_2} < 1 - (\beta^{(\pi_2,\pi_3)}\pstate)_{\pi_2}= \mathcal{L}_{\qstate'}(1-\tbeta_{\pi_2}), 
	\end{equation} 
	from $q'_{\pi_2} = (\beta^{(\pi_2,\pi_3)}\pstate)_{\pi_2}<p_{\pi_2}$. Therefore, $\qstate'$ strictly thermomajorises $\qstate$ and the non-extremality of $\qstate$ then follows from Lemma~\ref{lemma:same_border}. 
	\item For $E_{\pi_3} > E_{\pi_1}$, comparison between $\qstate$ with $\border(\qstate) = (\pi_2,\pi_3,\pi_1)$ and $\qstate^{\prime} = \beta^{(\pi_3,\pi_1)}\beta^{(\pi_1,\pi_2)}\pstate$ gives the same result.
	\item If $E_{\pi_3} = E_{\pi_1}$, we always get $ \border(\qstate) = (\pi_3,\pi_2,\pi_1) $.
\end{enumerate}
Note that from Eqs.~\eqref{eq:slope_after_swap1}--\eqref{eq:slope_after_swap4}, cases 1 and 2 cover all possible ways of obtaining $ \border(\beta^{\pi_{1},\pi_{3}}\pstate) \neq S_{1,3}(\border(\pstate)) $.

For the general case of $\pstate\in\probspace^d$, if $\border(\qstate =\beta^{(\pi_i,\pi_{i+c})}\pstate) \neq S_{i,i+c}(\border(\pstate))$, then the equivalent of $ \qstate' $ above can be chosen as follows:

\begin{enumerate}
	\item If $E_{\pi_i}<E_{\pi_{i+c}}$ and $ \pi_{i}\neq \border(\qstate)_{i+c} $, then $ \qstate^{\prime} = \beta^{(\pi_i,\pi_{i+c})}\beta^{(\pi_{i+c-1},\pi_{i+c})}\pstate \succ_\beta \qstate$.
	Although $ \border(\qstate)\neq\border(\qstate^{\prime}) $ in general, $ \qstate $ can always be obtained from $ \qstate^{\prime} $ by partial level thermalization between $ \pi_{i} $ and $ \pi_{i+c-1} $. To see this, notice that $ q_{\pi_{k}} = q^{\prime}_{\pi_{k}}$,  $\forall k\neq i,i+c-1 $, i.e. 
	\begin{equation}\label{eq:lemma6_proof_eq2}
		q_{\pi_{i}} + q_{\pi_{i+c-1}} = q^{\prime}_{\pi_{i}} + q^{\prime}_{\pi_{i+c-1}}.
	\end{equation}  
	Using the same argument to Eq.~\eqref{eq:lemma6_proof_eq1}, $ q^{\prime}_{\pi_{i}}>q_{\pi_{i}} $, and thus
	\begin{align}
		\slope(\qstate^{\prime})_{\pi_{i}} &> \slope(\qstate)_{\pi_{i}} > \slope(\qstate)_{\pi_{i+c-1}} > \slope(\qstate^{\prime})_{\pi_{i+c-1}}.
	\end{align}
	Combining with Eq.~\eqref{eq:lemma6_proof_eq2}, $ \qstate $ is obtained $ \qstate^{\prime} $ by partial thermalization between levels $ i $ and $ i+c-1 $. 
	 
	\item If $E_{\pi_i}>E_{\pi_{i+c}}$ and $ \pi_{i+c}\neq\border(\qstate)_{i} $, then $ \qstate^{\prime} = \beta^{(\pi_{i+c},\pi_{i})}\beta^{(\pi_{i},\pi_{i+1})}\pstate \succ_\beta \qstate$. Likewise, $ q_{\pi_{k}} = q^{\prime}_{\pi_{k}}$,  $\forall k\neq i+1,i+c $, and $ q_{\pi_{i+1}} + q_{\pi_{i+c}} = q^{\prime}_{\pi_{i+1}} + q^{\prime}_{\pi_{i+c}}$. With 
	\begin{align}
		\slope(\qstate^{\prime})_{\pi_{i+1}} &> \slope(\qstate)_{\pi_{i+1}} > \slope(\qstate)_{\pi_{i+c}} > \slope(\qstate^{\prime})_{\pi_{i+c}},
	\end{align}
	$ \qstate $ is obtained $ \qstate^{\prime} $ by partial thermalization between levels $ i+1 $ and $ i+c $.
	
	\item If $ E_{\pi_{i}} = E_{\pi_{i+c}} $, we always get $ \pi_{i} = \border(\qstate)_{i+c} $.
\end{enumerate}
Therefore, the state $ \qstate $ with order $\border(\qstate) \neq S_{i,i+c}(\border(\pstate))$ is always non-extremal in $ \clSeto(\pstate) $.

\subsection{Each $ \beta $-order has at least one $ \clSeto $ vertice: proof of Lemma~\ref{lemma:ext_point_for_each_order}}\label{appendix:proof_ext_point_for_each_order}

In this section, we show that it is impossible to have a $ \beta $-order that has no vertice of $ \clSeto $. Lemma \ref{lemma:ext_point_for_each_order} hints the lower bound scaling of the number of extreme points in worst cases. 

\begin{proof}
To start, we construct a series of sets $ S_d\subset S_{d-1}\subset\cdots\subset S_0 $ defined as follows:
	\begin{itemize}
		\item $ S_0 = \{\qstate\vert \qstate\in\clSeto(\pstate)\ \text{and}\ \border(\qstate) = \psi\} $, 
		\item $ S_j = \{\qstate\vert\qstate\in S_{j-1}\ \text{and}\ q_{\psi_j}\geq q^\prime_{\psi_j}, \forall \qstate^\prime\in S_{j-1}\} $ for $ 1\leq j\leq d $.
	\end{itemize}
	In other words, $ S_1 $ is the set of states in $ \clSeto(\pstate) $ with a specific $ \beta $-order $ \psi $, and with maximal $ \psi_1 $ population $ r_{\psi_1} $. Likewise, $ S_2 $ is a subset of $ S_1 $, having \textit{additionally} the maximal $ \psi_2 $ population, and so on. Note that $ S_d $ always has a single element for each fixed choice of $ \psi $, which we denote as $ \rstate $~\footnote{It should be noted at this point that $ \rstate $ may not be the only extreme point that has $ \beta $-ordering $ \psi $ -- there could be states $ \rstate' $ of the same $ \beta$-order, where the first elbow is lower that of $ \rstate $, and the second elbow higher. Such states are not, however, contained in $ S_d $.}. 
	
	Suppose there is no extreme point of $ \clSeto(\pstate) $ corresponding to $ \psi $. Then $ \rstate $ can be written as a strict convex combination of extreme states $ \estate^{(i)} $, i.e.
	\begin{equation}\label{key}
	\rstate = \sum_i 	p_i	\estate^{(i)} , \qquad p_i \in (0,1).
	\end{equation}
Starting from $ j=1 $, check the following:
	\begin{enumerate}
		\item For $ j > 1 $, we have $ e^{(i)}_{\psi_k} = r_{\psi_{k}} $, $ \forall k< j $ from the last iteration. For $ j=1 $, we do not need any condition yet. 
		
		\item If $ e^{(i)}_{\psi_j}>r_{\psi_j} $ for some $ i $, a state $ \estate^\prime = \vec\beta\estate^{(i)} $ with $ \border(\estate^\prime) = \psi $ and $ e^\prime_{\psi_k} = e^{(i)}_{\psi_k} $, $ \forall k\leq j $ can be found. Let us show how to do this:
		\begin{itemize}
			\item If $ \slope(\estate^{(i)})_{\psi_{j-1}} = \slope(\rstate)_{\psi_{j-1}} \geq \slope(\estate^{(i)})_{\psi_{j}} >\slope(\rstate)_{\psi_{j}}$, we can simply thermalise all the levels $ \psi_{k} $, $ \forall k>j $ of $ \estate^{(i)} $ to have the same slope, which is smaller than $ \slope(\estate^{(i)})_{\psi_{j}} $. Then we obtain the desired state $ \estate^{\prime} $, since levels with degenerate slopes -- all $ \psi_{k>j} $ in this case -- can be permuted within themselves in the $ \beta $-order. 
			\item If $ \slope(\estate^{(i)})_{\psi_{j}}> \slope(\rstate)_{\psi_{j-1}} >\slope(\rstate)_{\psi_{j}}$\footnote{Requiring $ \slope(\rstate)_{\psi_{j-1}} >\slope(\rstate)_{\psi_{j}} $ is always possible by putting all the levels having the same slope $ \slope(\rstate)_{\psi_{j}} $ to come after $ \psi_{j} $ in the order $ \psi $.}, we can first reduce $ e^{(i)}_{\psi_{j}} $ by partially thermalizing with populations of levels $ \psi_{k>j} $ until $ \slope(\rstate)_{\psi_{j-1}} \geq e^{\prime}_{\psi_{j}} > \slope(\rstate)_{\psi_{j}} $. Then, as in the previous case, thermalizing all the levels $ \psi_{k>j} $ will give $ \estate^{\prime} $.
		\end{itemize}
		
		However, such $ \estate^{\prime} $ satisfies $ \estate^\prime\in S_{j-1} $ and $ e^\prime_{\psi_j}>r_{\psi_j} $, which contradicts the assumption that $ \rstate \in S_j $.
		\item If $ e^{(i)}_{\psi_j}\leq r_{\psi_j} $ for all $ i $, from convexity of the combination, $ e^{(i)}_{\psi_j}=r_{\psi_j} $ for all $ i $. 
		Proceed to $ j\rightarrow j+1 $.
	\end{enumerate}
If $ e^{(i)}_{\psi_j} = r_{\psi_j} $ for all $ i $ and $ j $, $ \estate^{(i)} = \rstate $, which contradicts the assumption that $ \rstate $ is not extremal.

\end{proof}

\section{Full characterisation of elementary thermal operations for $d=3$ and the proof of Theorem.~\ref{thm:3dETO}}\label{appendix:3dETO_lemmas}
%
%
%
To prove Thm.~\ref{thm:3dETO}, we first establish three preliminary results for $ d=3 $. 
\begin{itemize}
	\item Remark~\ref{rmk:two_swap_TO_ext} identifies two-neighbouring-swap series that generate $\clSto$ extreme points, which are uniquely extremal for $\clSeto$ in their orders by Cor.~\ref{corollary:TO-ETO_ext}. This remark is used repeatedly to prove the other two lemmas.
	\item Lemma~\ref{rmk:only_one_nn} shows that non-neighbouring swaps cannot be used in series with more than one swap. 
	\item Lemma~\ref{rmk:len_max_3} sets the maximum length of swap series to be three. 
\end{itemize}
The set of candidates for extreme points is achieved after ruling out all the swaps yielding provably non-extreme states.

\begin{remark}\label{rmk:two_swap_TO_ext}
	For $\pstate\in\probspace^3$, 
	\begin{enumerate}
		\item the states $\beta^{(2,3)}\beta^{(1,2)}\pstate$ and  $\beta^{(2,3)}\beta^{(1,3)}\pstate $  are extreme points of $\clSto(\pstate)$, if $\border(\pstate) = (2,1,3)$ or $(3,1,2)$;
		\item the state $\beta^{(1,3)}\beta^{(2,3)}\pstate$ is an extreme point of $\clSto(\pstate)$, if $\border(\pstate) = (1,2,3)$ or $(3,2,1)$; and
		\item the state $\beta^{(1,2)}\beta^{(2,3)}\pstate$ is an extreme point of $\clSto(\pstate)$, if $\border(\pstate) = (1,3,2)$ or $(2,3,1)$.
	\end{enumerate}  
\end{remark}

\begin{proof}
	Direct calculation gives
	\begin{align}
		\beta^{(2,3)}\beta^{(1,2)} = \begin{pmatrix} 1-\Delta_{12} & 1 & 0\\ \Delta_{12} - \Delta_{13} & 0 & 1\\ \Delta_{13} & 0 & 0\end{pmatrix}, \quad
		\beta^{(2,3)}\beta^{(1,3)} = \begin{pmatrix} 1-\Delta_{13} & 0 & 1\\ \Delta_{13} & 1-\Delta_{23} & 0\\ 0 & \Delta_{23} & 0\end{pmatrix},\label{eq:ext_biplanar_matrices}\\
		\beta^{(1,2)}\beta^{(2,3)} = \begin{pmatrix} 1-\Delta_{12} & 1-\Delta_{23} & 1\\ \Delta_{12} & 0 & 0\\ 0 & \Delta_{23} & 0\end{pmatrix}, \quad
		\beta^{(1,3)}\beta^{(2,3)} = \begin{pmatrix} 1-\Delta_{13} & \Delta_{23} & 0\\ 0 & 1-\Delta_{23} & 1\\ \Delta_{13} & 0 & 0\end{pmatrix}.\nonumber
	\end{align}
	By using the algorithm in Def. 6 of \cite{Mazurek19_channels}, one can verify that the above channels are \emph{biplanar extreme points} of the set of thermal processes. 
	According to Thm. 4 in \cite{Mazurek19_channels}, such biplanar extremal channels generate extreme points of $\clSto(\pstate)$, when the initial state corresponds to a particular $ \beta $-ordering that can be found in the process of decomposing the graph structure of the channel matrix. Performing this procedure according to \cite{Mazurek19_channels} reveals that for $ \beta^{(2,3)}\beta^{(1,2)}  $ and $ \beta^{(2,3)}\beta^{(1,3)} $, the relevant input state $ \beta $-order is given by $ (2,1,3) $ and $ (3,1,2) $; similarly for statements 2 \& 3 in remark.
\end{proof}

\begin{lemma}\label{rmk:only_one_nn}
	
	Given $\pstate\in\probspace^3$ with $\border(\pstate) = (\pi_1,\pi_2,\pi_3)$, $\vec{\beta}\pstate$ is extremal for $\clSeto(\pstate)$ only if i) $\vec{\beta}$ is always neighbouring when applied to $\pstate$ or ii) $\vec{\beta} = \beta^{(\pi_1,\pi_3)}$.
	
\end{lemma}
\begin{proof}
	To prove this, we need to show that i) a neighbouring swap following a non-neighbouring swap produces non-extreme point and ii) a non-neighbouring swap following a neighbouring one also yields a non-extreme point. Since only extreme points are of our interest, using Lemma~\ref{lemma:order_swap_lemma}, we can safely assume that $ \border(\beta^{(\pi_{i},\pi_{j})}\pstate) = S_{i,j}(\pi) $ for any $ \pi = \border(\pstate) $ and $ i,j $.
	We tackle each problem by further dividing cases.

	\emph{Case i-(a):} $\border = (1,2,3)$ or $(3,2,1)$ experiencing a neighbouring swap followed by a non-neighbouring swap. Swap $ \beta^{(2,3)}\beta^{(1,2)} $ produces final states with order $ (3,1,2) $ or $ (2,1,3) $. From Remark~\ref{rmk:two_swap_TO_ext}, these orders have a unique extreme point produced from $ \beta^{(1,3)}\beta^{(2,3)} $. Another swap gives \begin{equation}\label{key}
		\rstate = \beta^{(1,2)}\beta^{(2,3)}\pstate = \begin{pmatrix}
			(1-\Delta_{12})p_1+(1-\Delta_{23})p_2 + p_{3} \\ \Delta_{12}p_1\\ \Delta_{23}p_{2}
		\end{pmatrix},
	\end{equation} 
	whereas two consecutive neighbouring swaps give
	\begin{equation}\label{key}
		\qstate = \beta^{(1,3)}\beta^{(1,2)}\pstate = \begin{pmatrix}
			(1-\Delta_{12})(1-\Delta_{13})p_1+(1-\Delta_{13})p_2 + p_{3} \\ \Delta_{12}p_1\\ \Delta_{13}(1-\Delta_{12})p_{1}+\Delta_{13}p_{2}
		\end{pmatrix},
	\end{equation}
	with order $ (2,3,1) $ or $ (1,3,2) $ depending on the initial order. 
	Note that $ r_{2} = q_{2} $ and 
	\begin{align}
		q_{3} - r_{3} = (1-\Delta_{12})\Delta_{23}(\Delta_{12}p_{1} - p_{2}),
	\end{align}
	i.e. $ q_{3} \geq r_{3} $ if initial $ \pi = (1,2,3) $, and $ q_{3} \leq r_{3} $ if initial $ \pi = (3,2,1) $. Either way, $ \qstate\succ_\beta\rstate $ while $ \border(\qstate) = \border(\rstate) $ and thus $ \rstate $ cannot be extremal from Lemma~\ref{lemma:same_border}.
	
	\emph{Case i-(b):} $\border = (1,2,3)$ or $(3,2,1)$ experiencing a non-neighbouring swap followed by a neighbouring swap. We can in fact prove that non-neighbouring swap already always produces non-extreme points. Compare two states 
	\begin{eqnarray}
		\rstate = \beta^{(1,3)}\pstate &=& \begin{pmatrix}
			(1-\Delta_{13})p_1+p_3\\ p_2\\ \Delta_{13}p_1
		\end{pmatrix},\label{eq:non-neighbouring-non-extremal1}\\
		\qstate = \beta^{(2,3)}\beta^{(1,3)}\beta^{(1,2)}\pstate
		&=& \begin{pmatrix}
			(1-\Delta_{12})(1-\Delta_{13})p_1 + (1-\Delta_{13})p_2+p_3\\ \Delta_{12}(1-\Delta_{13})p_1 + \Delta_{13}p_2\\ \Delta_{13}p_1
		\end{pmatrix},\label{eq:non-neighbouring-non-extremal2}
	\end{eqnarray}
	with $\border(\rstate) = \border(\qstate) = (3,2,1)$ or $ (1,2,3) $ depending on the initial $ \pi $. 
	Then we may observe the following:
	\begin{align}
		r_3 &=q_3, \label{eq:non-neighbouring-non-extremal3}\\ 
		r_2 - q_2 &= (1-\Delta_{13})(p_2-\Delta_{12}p_1)\leq 0,  \ \text{for}\ \pi = (1,2,3), \label{eq:non-neighbouring-non-extremal4}\\
		r_2 - q_2 &= (1-\Delta_{13})(p_2-\Delta_{12}p_1)\geq 0,  \ \text{for}\ \pi = (3,2,1), \label{eq:non-neighbouring-non-extremal5}
	\end{align}
	i.e. $\qstate \succ_{\beta} \rstate$ and thus $\rstate$ is not extremal for $\clSeto(\pstate)$.

	\emph{Case ii-(a):} $\border = (1,3,2)$ or $(2,3,1)$ experiencing a neighbouring swap followed by a non-neighbouring swap. The two possible neighbouring swaps for these initial orders are $\beta^{(2,3)}$ and $\beta^{(1,3)}$. First, consider the neighbouring swap $\beta^{(2,3)}$: this modifies the order into $(1,2,3)$ or $(3,2,1)$ and \emph{Case i-(b)} forbids a non-neighbouring swap to come next. The other neighbouring and non-neighbouring swap pair $\beta^{(2,3)}\beta^{(1,3)}$ gives output orders $(2,1,3)$ or $(3,1,2)$. But Remark~\ref{rmk:two_swap_TO_ext} states that $\beta^{(1,2)}\beta^{(2,3)}\pstate$ is a unique extreme point for that output order. 
	
	\emph{Case ii-(b):} $\border = (1,3,2)$ or $(2,3,1)$ experiencing a non-neighbouring swap followed by a neighbouring swap. After a non-neighbouring swap, $\border(\beta^{(1,2)}\pstate) = (2,3,1)$ or $(1,3,2)$. If a following neighbouring swap is $\beta^{(1,3)}$, the resulting orders are $(2,1,3)$ or $(3,1,2)$, which again cannot be extremal from the Remark~\ref{rmk:two_swap_TO_ext}. The remaining possibility is to apply $\beta^{(2,3)}$, which results in 
	\begin{align}
		\rstate &= \beta^{(2,3)}\beta^{(1,2)}\pstate=\begin{pmatrix}
			(1-\Delta_{12})p_1+p_2 \\ (\Delta_{12}-\Delta_{13})p_1+p_3\\ \Delta_{13}p_1
		\end{pmatrix}.
	\end{align}
	Compare this with a state having the same $\beta$-order $ \border(\rstate) = \border(\qstate) =  (3,2,1) $ or $ (1,2,3) $,
	\begin{equation}
		\qstate = \beta^{(1,2)}\beta^{(1,3)}\pstate\label{eq:rmk6_klkm}
		=\begin{pmatrix}
			(1-\Delta_{12})(1-\Delta_{13})p_1 + p_2 + (1-\Delta_{12})p_3 \\ \Delta_{12}(1-\Delta_{13})p_1 + \Delta_{12}p_3\\ \Delta_{13}p_1
		\end{pmatrix}.
	\end{equation}
	Then 
	\begin{align}
		q_1 - r_1 &= -\Delta_{13}(1-\Delta_{12})p_1+(1-\Delta_{12})p_3\leq0,\label{eq:ineq_rmk6_3}
	\end{align}
	for initial $\pi = (2,1,3)$ and becomes positive for $(3,1,2)$. Plus, $ q_3 = r_3 $, which in turn gives $\qstate\succ_{\beta}\rstate$; thus, $\rstate$ is not extremal for $\clSeto(\pstate)$.
	
	\emph{Case iii-(a):} $\border = (2,1,3)$ or $(3,1,2)$ experiencing a neighbouring swap followed by a non-neighbouring swap. If the first neighbouring swap is $\beta^{(1,2)}$, the output $ \beta $-order becomes $ (1,2,3) $ or $ (3,2,1) $, which does not allow non-neighbouring swap to follow as stated in \emph{Case i-(b)}. The other series, $\beta^{(1,2)}\beta^{(1,3)}$ outputs orders $(1,3,2)$ or $(2,3,1)$, but these orders have unique extreme points for $\clSeto$ given by Remark~\ref{rmk:two_swap_TO_ext}.
	
	\emph{Case iii-(b):} $\border = (2,1,3)$ or $(3,1,2)$ experiencing a non-neighbouring swap followed by a neighbouring swap. Two candidate series are $\beta^{(1,3)}\beta^{(2,3)}$ and $\beta^{(1,2)}\beta^{(2,3)}$, which respectively produces orders $(1,3,2)$ and $(3,2,1)$ when applied to an initial state with $\pi = (2,1,3)$; $(2,3,1)$ and $(1,2,3)$ when applied to $\pi = (3,1,2)$. All output states obtained here have different unique extreme points for $\clSeto$ given in Remark~\ref{rmk:two_swap_TO_ext}, and the states generated by the considered swaps are therefore non-extremal.
	
	We exhausted all possible cases and none of the swaps can create an extreme point of $\clSeto(\pstate)$. 
	
\end{proof}

\begin{lemma}\label{rmk:len_max_3}
	For $\pstate\in\probspace^3$ having non-degenerate energy levels, $\prod_{i=1}^{l}\beta_i\pstate$ can be extremal in $\clSeto(\pstate)$ only when $l\leq 3$. 
\end{lemma}
\begin{proof}
	We prove this lemma by showing $l=4$ swaps always produce non-extreme states. From Lemma~\ref{rmk:only_one_nn}, non-neighbouring swaps do not need to be considered. Then the remaining two final states are $\qstate^{(1,2,1,2)}$ and $\qstate^{(2,1,2,1)}$.
	
	\emph{Case i:} $\border(\pstate) = (1,2,3)$ or $(3,2,1)$. There are only three distinct $\beta$-swaps for three-dimensional systems. Then length-four series always include a repetition of the same swap, and Thm.~\ref{thm:nice_order} forbids such series to make extreme states. 
	
	\emph{Case ii:} $\border(\pstate) = (1,3,2)$ or $(2,3,1)$. Corresponding swaps are $\beta^{(1,3)}\beta^{(2,3)}\beta^{(1,2)}\beta^{(1,3)}$ or $\beta^{(2,3)}\beta^{(1,3)}\beta^{(1,2)}\beta^{(2,3)}$. From Remark~\ref{rmk:swap_series_eq}, the second swap becomes
	\begin{equation}
		\beta^{(2,3)}\beta^{(1,3)}\beta^{(1,2)}\beta^{(2,3)} = \beta^{(1,2)}\beta^{(1,3)}\left(\beta^{(2,3)}\right)^2,
	\end{equation}
	and thus the channel itself is non-extremal. If the first swap is applied, the resulting orders are $(2,1,3)$ for $\pi = (1,3,2)$ and $(3,1,2)$ for $\pi = (2,3,1)$. Different states obtained from Remark~\ref{rmk:two_swap_TO_ext} are known to be uniquely extremal for these orders, rendering the states after the first swap non-extremal.  
	
	\emph{Case iii:}	$\border(\pstate) = (2,1,3)$ or $(3,1,2)$. Corresponding output orders are $(1,3,2)$ and $(3,2,1)$ for $\pi = (2,1,3)$;  $(2,3,1)$ and $(1,2,3)$ for $\pi = (3,1,2)$. Again, Remark~\ref{rmk:two_swap_TO_ext} provides different states that are uniquely extremal in the orders above.
	
\end{proof}

When there is a degeneracy in energy levels, e.g., $E_k=E_l$, the equality $\beta^{(k,l)}\beta^{(k,m)} =  \beta^{(l,m)}\beta^{(k,l)}$ let us rearrange the swap to make $\beta^{(k,l)}$ appear only at the end of the series. From $(\beta^{(k,l)})^2 = \mathbb{1}$, we can always make $\beta^{(k,l)}$ never appear or appear only once in $\vec\beta$. After this reduction, Lemma~\ref{rmk:len_max_3} also holds for degenerate energy systems.

Ruling out series with length $>1$ having non-neighbouring swaps (Lemma~\ref{rmk:only_one_nn}) and series with length $l>3$ (Lemma~\ref{rmk:len_max_3}), we are only left with $\beta^{(\pi_1,\pi_3)}$ and all-neighbouring swaps with length $\leq 3$. $\Theta_\mathrm{ETO}^{(3)}(\pstate) \cup \Xi_\mathrm{ETO}^{(3)}(\pstate)$ is the collection of all states after such swaps, and this proves $\mathrm{Ext}(\clSeto)~\subset~\Theta_\mathrm{ETO}^{(3)}(\pstate) \cup \Xi_\mathrm{ETO}^{(3)}(\pstate)$ part of Thm.~\ref{thm:3dETO}. To prove the other part, notice that elements of $\Xi_\mathrm{ETO}^{(3)}(\pstate)$ all have the same $\beta$-order $(\pi_3,\pi_2,\pi_1)$, while the ones in $\Theta_\mathrm{ETO}^{(3)}(\pstate)$ all have distinct $\beta$-orders different from $(\pi_3,\pi_2,\pi_1)$. Since each $\beta$-order possesses at least one extreme state (Lemma~\ref{lemma:ext_point_for_each_order}), all $\Theta_\mathrm{ETO}^{(3)}(\pstate)$ elements should be extremal, which concludes the proof.

\section{Simplifications for initial orderings monotonic in energy levels: proof of Theorem~\ref{thm:nice_order}}\label{app:special_ordering}

The technical proof of Thm.~\ref{thm:nice_order} can be sketched in the following steps:
\begin{itemize}
	\item Firstly, we introduce a specific series of $\beta$-swaps that transforms an initial $\beta$-ordering to a target ordering (Def.~\ref{def:standard_formation}). 
	\item This structure, which we refer to as the \emph{standard formation},	is then shown to be equivalent to any $\beta$-swap series, where i) all swaps are \emph{neighbouring} to initial states of the form Eq.~\eqref{eq:mono_order} and ii) each swap is applied \emph{at most once} (Lemma~\ref{lemma:standard_form}). 
	\item Finally, to prove Thm.~\ref{thm:nice_order}, we show that whenever the initial $ \beta $-ordering is monotonic in energy, then a transformation that is not according to a standard formation always leads to a non-extreme state. 
	This Lemma allows us to conclude that the number of extreme points for such a $ \clSeto(\pstate) $ is at most $ d! $, similar to that of $ \clSto (\pstate)$.
\end{itemize}

\begin{definition}[Standard formation]\label{def:standard_formation}
	 Given a tuple of $d$-dimensional $ \beta $-orderings $ (\pi,\pi') $, a standard formation is a $\beta$-swap series $\vec{\beta}_{\rm sf}$ that transforms an initial state $ \pstate $ with an order $\pi$ into some final state $ \pstate' $ having an order $\pi^\prime$, 
	with the construction below:
\begin{enumerate}
	\item Set an initial index of $ j=1 $, and identify $ m $ such that $\pi^\prime_j = \pi_{m}$. If $m=j$, define $\vec\beta^{(j)}$ as an identity. Otherwise, since $ \pi_{1},\cdots,\pi_{m-1} $ are already occupied by  $ \pi^{\prime}_{1},\cdots,\pi^{\prime}_{m-1} $, we get $ m>j $. Then, define a swap-series $\vec\beta^{(j)} = \beta^{(\pi_{j},\pi_{m})}\beta^{(\pi_{j+1},\pi_{m})}\cdots\beta^{(\pi_{m-1},\pi_{m})}$. Note that these swaps are always neighbouring when applied to a state initially having the order $\pi$, due to Eq.~\eqref{eq:order_after_swap1}. After $ j=1 $ round, this swap series will take the initial ordering $ \pi $ to the new ordering $ (\pi_1',\pi_1,\cdots,\pi_{m-1},\pi_{m+1},\cdots,\pi_d) $ if $ \vec{\beta}^{(1)}\neq\mathbb{1} $.
	\item Iterate the above step for $ j=2,\cdots,d-1 $, defining $\lbrace \vec\beta^{(j)}\rbrace_{j=1}^{d-1}$. 
\end{enumerate} 
	The standard formation series is then simply the concatenation
\begin{equation}
	\vec{\beta}_{\rm sf} = \vec\beta^{(d-1)} \cdots \vec\beta^{(2)} \vec\beta^{(1)}.
\end{equation}
By construction, there is no repetition of a swap in this series and they are all neighbouring when applied to an initial state with the order $\pi$. 
\end{definition}

\begin{mybox}{An example of the standard formation}
	For the orderings $ \pi = 1234 $ and $ \pi'=4231 $, the standard formation is $ \vec{\beta}_{\rm sf} = \vec{\beta}^{(3)}\vec{\beta}^{(2)}\vec{\beta}^{(1)} $, where
	\begin{align}\label{key}
		\vec{\beta}^{(1)} &=\beta^{(1,4)}\beta^{(2,4)}\beta^{(3,4)},\qquad
		\vec{\beta}^{(2)} =\beta^{(1,2)},\qquad
		\vec{\beta}^{(3)} =\beta^{(1,3)}. 
	\end{align}
	The intermediate $ \beta $-orderings given by this process are
	\begin{equation}\label{key}
		\pi = 1234 \quad \xrightarrow[\vec{\beta}^{(1)}]{}~ \quad 4123 \quad \xrightarrow[\vec{\beta}^{(2)}]{}~ \quad 4213 \quad
		\xrightarrow[\vec{\beta}^{(3)}]{}~\quad \pi ' = 4231.
	\end{equation}
\end{mybox}

This formation has a nice property, namely all swaps in each block $\vec{\beta}^{(j)}$ acts on level $\pi^\prime_j$ and the ones in $\vec{\beta}^{(k)}$ for any $k>j$ does not act on level $\pi^\prime_j$. In the lemma below, we illustrate how certain classes of swap series, which turns out to be the ones producing extreme states, can always rearranged into a standard formulation.

\begin{lemma}\label{lemma:standard_form}
	Given an initial state $ \pstate $ with ordering $ \pi $ monotonic in energy, denote a $\beta$-swap series 
	\begin{equation}\label{key}
		\vec\beta = \prod_{i} \beta_i,
	\end{equation} 
and $ \pi' $ to be the final $ \beta $-ordering of the state $\pstate' = \vec{\beta}\pstate $. If $ \vec{\beta} $ is such that: \\
1) each $ \beta_i $ is a distinct swap, and  \\
2) when applied to $ \pstate $, is always a neighbouring swap, \\
then $ \vec{\beta} $ can always be expressed in the form of a standard formation for $ (\pi,\pi') $.
\end{lemma}
\begin{proof}
	
	We prove this by induction. Suppose the above lemma is true for $\pstate\in\probspace^{d-1}$. The first goal is to prove this for initial ordering $\pi =  (1,2,\cdots,d)$ and $ \vec{\beta} $ that satisfies the conditions in the statement of the lemma. By identifying the first and the last swaps acting on the level $d$, one can decompose $ \vec{\beta} $ into 
	\begin{equation}\label{key}
		\vec{\beta} = \vec{\beta}^{(\text{Post})}\vec{\beta}^{(d-\text{rel})}\vec{\beta}^{(\text{Pre})}.
	\end{equation}
	Here, $\vec{\beta}^{(\text{Pre})}$ are the swaps coming before the first swap acting on the level $d$, and $\vec{\beta}^{(\text{Post})}$ are the ones after the last swap acting on the level $d$. 
	
	\emph{Case i:} $d = \pi^\prime_{m}$, where $ m\neq 1 $ is the position of level $ d $ in the $ \beta $-ordering of the final state $ \pstate' $. To make a rearrangement of swaps, we first remark a few points using sets $A = \{i\vert\slope_i>\slope_d\}$ and $B = \{i\vert\slope_i<\slope_d\}$. We will update these sets after each swap. At the beginning, there is no element in $B$ and all the other levels except $ d $ are in the set $A$. Next, we note the following:
	\begin{enumerate}
		\item Swapping $i\in A$ and $j\in B$ is not allowed, since they are non-neighbouring.
		\item Any $i\in A$ can move to $B$ only when $ \beta^{(i,d)} $ is implemented.
		\item Since initially $ B = \emptyset $, any given level either stays in $A$ at all times, or it moves to $B$ at some point and remains so thereafter. This comes from the restriction that in order for $ i $ to move between $ A $ and $ B $, the swap $  \beta^{(i,d)}$ must be used.
		\item If $k,l \in A$ when $\beta^{(k,l)}$ is applied, $\beta^{(k,l)}$ precedes $\beta^{(k,d)}$ and $\beta^{(l,d)}$ when they exist in $\vec\beta$.
	\end{enumerate}
	  The sets $A$ and $B$ after the whole transformation is determined by the target $ \beta $-ordering $ \pi' $, where we denote them as $A_f = \{\pi^\prime_i\vert i<m\}$ and $B_f = \{\pi^\prime_i\vert i>m\}$. 

  	Given the constraints above, we know that $ \vec{\beta} $ describes a special process. See Fig.~\ref{fig:lemma8_illustration}, for instance, for a visualization of this operation. $A$ and $B$ are separated by level $d$ (point 1 above). Starting from $B = \emptyset$, some elements $i\in A$ are transferred to $B$ whenever $ \beta^{(i,d)} $ is implemented. 
  	Once this happens, $ i $ cannot go back to $A$, since it would require the repetition of $  \beta^{(i,d)} $ to do so. Visually, this is understood by saying that the bar representing level $d$ in Fig.~\ref{fig:lemma8_illustration} is penetrable from the left only. At the end, $A = A_f$ and $B=B_f$, where elements of $A_f$ never passed through level $d$ (point 3 above). Lastly, if $k,l\in A$, then they have not experienced a swap with level $d$ yet, explaining the point 4 above.  

	\begin{figure}[h!]\centering
		\includegraphics[width = 0.88\columnwidth]{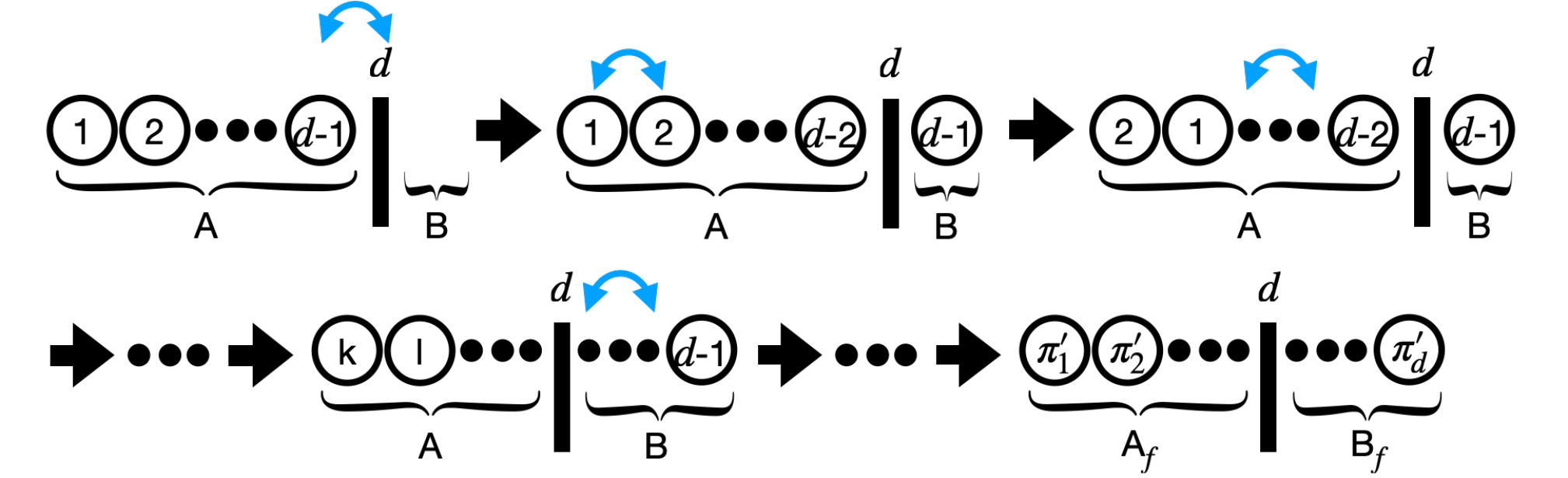}
		\caption{
			Illustration of changes in $\beta$-order starting from $\pi = (1,2,\cdots,d)$ undergoing a swap series $\vec\beta$, where it is assumed that $ \vec\beta $ is always neighbouring and allows no repetition of a particular swap. 
			Level $d$ is depicted as a bar, dividing the levels into sets $A$ and $B$. For an initial state whose $ \beta $-order is monotonically increasing in energy, $ A = \lbrace 1, \cdots, d-1\rbrace $ while $ B = \emptyset$. Three different types of swaps are possible: i) swapping between level $d$ and its neighbouring element which is in $A$ (1st swap above), ii) swapping among neighbouring levels in $A$ (2nd, 3rd steps above), and iii) swaps among levels in $B$ (4th explicit step above). When the type ii) and iii) swaps occur, $A$ and $B$ remain the same. On the other hand, type i) swaps move one element of $A$ into $B$. The process of elements going from $B$ to $A$ is forbidden due to the no-repetition constraint. At the end of applying $\vec\beta$, $A$ and $B$ becomes $A_f$ and $B_f$. The elements of $A_f$ never experience a swap with level $d$, while the elements of $ B_f $ experienced exactly once swapping with level $ d $.}
		\label{fig:lemma8_illustration}
	\end{figure}
  	
  	The next step we want to show is that w.l.o.g., a rearrangement, where $ \vec{\beta}^{\rm (Pre)} $ contains all $ \beta$-swaps that are part of $ \vec{\beta} $, acting on elements $ i\in A_f $, but not involving level $ d $, is possible.
  	To do so, we identify all swaps $\beta^{(k,l)}$ in $\vec{\beta}^{(\text{Post})}\vec{\beta}^{(d-\text{rel})}$  such that $k,l\in A$ at the time of swap. From the rightmost one, make a decomposition
  	\begin{equation}\label{key}
  		\vec{\beta}^{(\text{Post})}\vec{\beta}^{(d-\text{rel})} = \vec{\beta}^{(a)}\beta^{(k,l)}\vec{\beta}^{(b)}.
  	\end{equation}
  	Notice that $\vec{\beta}^{(b)}$ only contains swaps among $B\cup\{d\}$, which includes neither $k$ nor $l$. From Remark \ref{rmk:commute}, we then have that $ \beta^{(k,l)} $ and $  \vec{\beta}^{(b)}$ commute, i.e. $\beta^{(k,l)}\vec{\beta}^{(b)} = \vec{\beta}^{(b)}\beta^{(k,l)}$. Therefore, w.l.o.g., 
  	\begin{equation}\label{key}
  		\vec{\beta}^{(\text{Post})}\vec{\beta}^{(d-\text{rel})} \vec{\beta}^{(\text{Pre})} = \vec{\beta}^{(a)}\beta^{(k,l)}\vec{\beta}^{(b)} \vec{\beta}^{(\text{Pre})} = \vec{\beta}^{(a)}\vec{\beta}^{(b)} \beta^{(k,l)}\vec{\beta}^{(\text{Pre})}
  	\end{equation}$ \beta^{(k,l)} $ can be integrated into $ \beta^{(\text{Pre})} $ to update $ \beta^{(\text{Pre})} \rightarrow \beta^{(k,l)}\beta^{(\text{Pre})} $ and $ \vec{\beta}^{(\text{Post})}\vec{\beta}^{(d-\text{rel})} \rightarrow \vec{\beta}^{(a)}\vec{\beta}^{(b)} $.

 If $k\in A_f$, the existence of $\beta^{(k,l)}$ in $\vec\beta$ indicates $l\in A$ when the swap is applied. Hence, by repeating this until the end, all swaps $\beta^{(k,l)}$ with $k\in A_f$ are merged into $\vec{\beta}^{(\text{Pre})}$. 
 Since $\vec{\beta}^{(\text{Pre})}$ does not contain swaps acting on level $d$, it acts on at most $d-1$ levels and can be reordered in the standard formation by the assumption that the lemma holds for states in $ \probspace^{d-1} $. 
Until now, there is no swap acting on level $ d $ and thus $ d = \border(\vec{\beta}^{(\text{Pre})}\pstate)_{d} $. Then levels in $ B_{f} $, which should be swapped with $ d $ in $ \vec{\beta}^{(d-\text{rel})} $, occupy later $d-m$ slots in the $ \beta $-order: $ \border(\vec{\beta}^{(\text{Pre})}\pstate)_{m}, \border(\vec{\beta}^{(\text{Pre})}\pstate)_{m+1},\cdots,\border(\vec{\beta}^{(\text{Pre})}\pstate)_{d-1}$. By construction of the standard formation, $\vec{\beta}^{(\text{Pre})}$ then can be decomposed into $\vec{\beta}^{(\text{Pre})} = \vec{\beta}^{(B_f)}\vec{\beta}^{(A_f)}$ with $\vec{\beta}^{(B_f)}$ swapping only between $ B_{f} $ elements.

  	Finally, $\vec{\beta}^{(\text{Post})}\vec{\beta}^{(d-\text{rel})}\vec{\beta}^{(B_f)}$ consists of swaps among the levels in $B_f\cup \{d\}$ ($d-m+1< d$ elements). Again, by assumption this swap can be rearranged as a standard formation. Concatenating standardised series $\vec{\beta}^{(\text{Post})}\vec{\beta}^{(d-\text{rel})}\vec{\beta}^{(B_f)}$ and $ \vec{\beta}^{(A_f)} $, we obtain the standard formation for the entire series. \\

	\emph{Case ii:} $d = \pi^\prime_1$. The only difference here is that level $d$ swaps with every other level and $A_f$ is an empty group. Again, we locate $\beta^{(k,l)}$ such that  $\vec{\beta}^{(\text{Post})}\vec{\beta}^{(d-\text{rel})} = \vec{\beta}^{(a)}\beta^{(k,l)}\vec{\beta}^{(b)}$ from the rightmost swap. If $k,l\in A$ after $\vec{\beta}^{(b)}$, move $\beta^{(k,l)}$ to be included in $\vec{\beta}^{(\text{Pre})}$ as before. In addition, repeat this process for $ \vec{\beta}^{(d-\text{rel})} $ but starting from the leftmost swap $\beta^{(k,l)}$ with $k,l\neq d$ in $\vec{\beta}^{(d-\text{rel})} = \vec{\beta}^{(a)}\beta^{(k,l)}\vec{\beta}^{(b)}$. 
	Since we already moved all $k,l\in A$ swaps to go before $ \vec{\beta}^{(d-\text{rel})} $, at the point of swap $ \beta^{(k,l)} $ all $ k,l\in B $, which leads to the equality $\vec{\beta}^{(a)}\beta^{(k,l)} = \beta^{(k,l)}\vec{\beta}^{(a)}$ and enables $ \vec{\beta}^{(\text{Post})} \rightarrow \vec{\beta}^{(\text{Post})}\vec{\beta}^{(k,l)} $.
	After merging all such $\beta^{(k,l)}$ into $\vec{\beta}^{(\text{Pre})}$ or $\vec{\beta}^{(\text{Post})}$, we get
	\begin{equation}\label{key}
		\vec{\beta}^{(d-\text{rel})} =\beta^{(\delta_1,d)}\beta^{(\delta_2,d)}\cdots\beta^{(\delta_{d-1},d)},
	\end{equation} 
where $\delta = (\delta_1,\delta_2,\cdots,\delta_{d-1},d)$ is the $\beta$-order after $\vec{\beta}^{(\text{Pre})}$. Both $\vec{\beta}^{(\text{Pre})}$ and $\vec{\beta}^{(\text{Post})}$ act at most $d-1$ levels, and can be modified into the standard formation. Now to put the entire series into the standard formation, $\delta$ need to be rearranged. We do this starting from $ j=1 $.
	\begin{enumerate}
		\item Find $m$ such that $\delta_m = j$. If $m = j$, proceed to the last step. If not, previous iterations guarantee that $ \delta_{k}=k $, $ \forall k<j $, which leads to $ m>j $ and $ \delta_{m-1}>j $.
		Defining $L = \{\delta_n\vert n<m-1\}$ and $R = \{\delta_n\vert n>m\}$, we get 
		\begin{equation}\label{key}
			\vec{\beta}^{(d-\text{rel})} = \vec{\beta}^{(dL)}\beta^{(\delta_{m-1},d)}\beta^{(j,d)}\vec{\beta}^{(dR)},
		\end{equation}
	where $ \vec{\beta}^{dL(R)} $ denotes the series swapping $ d $ and elements of $ L(R) $.
		Moreover, from the standardization $\vec{\beta}^{(\text{Pre})} = \vec{\beta}^{(R)}\beta^{(j,\delta_{m-1})}\vec{\beta}^{(L)}$, where $\vec{\beta}^{(R)}$ does not act on levels $j$ and $\delta_{m-1}$ and we can rearrange it into  $\vec{\beta}^{(dR)}\vec{\beta}^{(\text{Pre})} = \beta^{(j,\delta_{m-1})}\vec{\beta}^{(dR)}\vec{\beta}^{(R)}\vec{\beta}^{(L)}$. From Remark~\ref{rmk:swap_series_eq}, $\beta^{(\delta_{m-1},d)}\beta^{(j,d)}\beta^{(j,\delta_{m-1})} = \beta^{(j,\delta_{m-1})}\beta^{(j,d)}\beta^{(\delta_{m-1},d)}$ since $ d>\delta_{m-1}>j $. 
		This procedure updates $\delta \rightarrow  (\cdots,\delta_{m-2},j,\delta_{m-1},\delta_{m+1},\cdots )$.
		\item Repeat the first step until $\delta_{j} = j$.
		\item Repeat the first and the second step with $ j\rightarrow j+1 $.
	\end{enumerate}
	At the end, one gets $ \delta = (1,2,\cdots,d-1) $ and $ \vec{\beta}^{(\text{Pre})}  = \mathbb{1} $. By standardizing $ \vec{\beta}^{(\text{Post})} $ and concatenating with $ \vec{\beta}^{(d-\text{rel})} $, the standard formation is obtained. 
	
	For $d=2$, the Lemma is trivially true. Thus by induction, the lemma is proved for $\pi = (1,2,\cdots,d)$. Following the same logic, this can also be proven for $\pi = (d,d-1,\cdots,1)$.

\end{proof}


Now we prove Thm.~\ref{thm:nice_order}.

\begin{proof}

We prove the Lemma for the case $\border(\pstate) = (1,2,\cdots,d)$ and argue that the proof also holds for $\border(\pstate) = (d,d-1,\cdots,1)$.

We first prove the only if statement of the lemma, as follows:
\begin{enumerate}
	\item We show that the repetition of any particular $ \beta $-swap always leads a to non-extreme state if the series is all-neighbouring. This is done by contradiction: suppose that $\vec\beta = \beta^{(k,l)}\vec{\beta}^\prime$, w.l.o.g. assuming $k<l$. Furthermore, assume $\vec{\beta}^\prime$ to be a series satisfying the only if part of the statement but contains $\beta^{(k,l)}$, causing $\beta^{(k,l)}$ to occur twice in  $\vec\beta$. From Lemma~\ref{lemma:standard_form}, $\vec{\beta}^\prime$ can be written in a standard formation, which reads
	\begin{equation}
		\vec\beta = \beta^{(k,l)}\vec{\beta}^\prime = \beta^{(k,l)}\vec{\beta}^{(\text{Post})}\beta^{(k,l)}\vec{\beta}^{(\text{Pre})}.
	\end{equation}   
	Notice that since $\vec\beta$ is an all-neighbouring swap, this implies that after $\vec\beta^\prime$, levels $k$ and $l$ should be neighbouring with $\slope(\vec\beta^\prime\pstate)_k\leq\slope(\vec\beta^\prime\pstate)_l$, which then implies  
	\begin{equation}
		\vec{\beta}^{(\text{Post})} = \vec{\beta}^{(\text{Irrel})} \left(\prod_{m_i\in M}\beta^{(m_i,k)}\right)\left(\prod_{m_i\in M}\beta^{(m_i,l)}\right),
	\end{equation} 
	for some set of levels $M \subset \{m\vert m<k<l\}$, 
	by construction of the standard formation. Here, $\vec{\beta}^{(\text{Irrel})}$ is a series that acts on neither $k$ nor $l$. Using Remarks~\ref{rmk:commute} and~\ref{rmk:swap_series_eq},
	\begin{align}
		\vec\beta &= \vec{\beta}^{(\text{Irrel})}\beta^{(k,l)}\prod_{m_{i}\in M}\left(\beta^{(m_i,k)}\beta^{(m_i,l)}\right)\beta^{(k,l)}\vec{\beta}^{(\text{Pre})}\nonumber\\ 
		&= \vec{\beta}^{(\text{Irrel})}   \left(\beta^{(k,l)}\right)^2\prod_{m_{i}\in M}\left(\beta^{(m_i,l)}\beta^{(m_i,k)}\right)  \vec{\beta}^{(\text{Pre})},\label{eq:rearrange_using_comm}
	\end{align}
	which always generates a non-extreme point from Remark~\ref{rmk:no_repetition}.
	\item Now we show that non-neighbouring swaps are also not allowed. 
	\begin{enumerate}
		\item \textit{For $ d=3 $}: from Eqs.~\eqref{eq:non-neighbouring-non-extremal1}--\eqref{eq:non-neighbouring-non-extremal5}, $ \beta^{(1,3)} $ yields non-extreme point of $ \clSeto(\pstate) $ with $ \border(\pstate) = (1,2,3) $ or $ (3,2,1) $. Also, Lemma~\ref{rmk:only_one_nn} forbids any other occasions having a non-neighbouring swap. 
		
		\item \textit{For $ d>3 $:} suppose that $\beta^{(k,m)}$ is the only non-neighbouring swap in the series $ \vec{\beta} $, i.e. 
				\begin{equation}\label{key}
						\vec{\beta} = \beta^{(k,m)}\vec{\beta}^{( \text{NS})} 
					\end{equation} and  \begin{equation}\label{key}
					\exists l\quad \text{s.t.}\quad \slope(\vec{\beta}^{(\text{NS})}\pstate)_k>\slope(\vec{\beta}^{(\text{NS})}\pstate)_l>\slope(\vec{\beta}^{(\text{NS})}\pstate)_m,
				\end{equation} after all-neighbouring series $ \vec{\beta}^{(\text{NS})} $. From the first part of the proof, $ \vec{\beta}^{(\text{NS})} $ also cannot include any repetition. Using Lemma~\ref{lemma:standard_form}, we rearrange $ \vec{\beta}^{(\text{NS})} $ into the standard formation. 
		\begin{enumerate}
			\item  $ \exists l $ satisfying $k<l<m$ or $k>l>m$: firstly, note that the procedures from Eqs.~\eqref{eq:non-neighbouring-non-extremal1}--\eqref{eq:non-neighbouring-non-extremal5} can be generalised for higher dimensions. If $\slope(\rstate)_k>\slope(\rstate)_l>\slope(\rstate)_m$ and $k<l<m$ for some $\rstate$, 
					\begin{align}
							(\beta^{(k,m)}\rstate)_{i} &= (\beta^{(l,m)}\beta^{(k,m)}\beta^{(k,l)}\rstate)_{i},\quad \forall i\neq k,l,\\
							(\beta^{(k,m)}\rstate)_{k} + (\beta^{(k,m)}\rstate)_{l} &= (\beta^{(l,m)}\beta^{(k,m)}\beta^{(k,l)}\rstate)_{k} + (\beta^{(l,m)}\beta^{(k,m)}\beta^{(k,l)}\rstate)_{l},
						\end{align}
					and
					\begin{equation}
							\slope(\beta^{(l,m)}\beta^{(k,m)}\beta^{(k,l)}\rstate)_{l} \geq \slope(\beta^{(k,m)}\rstate)_{l} \geq  \slope(\beta^{(k,m)}\rstate)_{k} \geq \slope(\beta^{(l,m)}\beta^{(k,m)}\beta^{(k,l)}\rstate)_{k}, 
						\end{equation}
					which implies that $\beta^{(k,m)}\rstate$ can be obtained from $\beta^{(l,m)}\beta^{(k,m)}\beta^{(k,l)}\rstate$ via partial thermalization of levels $ k $ and $ l $ (cf. $ d $-dimensional case for \ref{appendix:order_swap_lemma}), and thus not extremal in $ \clSeto(\rstate) $. Similarly, the result also holds when $ k>l>m $. By putting $ \rstate = \vec{\beta}^{(\text{NS})}\pstate $, the state $ \vec{\beta}\pstate $ is not extremal in $ \clSeto(\vec{\beta}^{(\text{NS})}\pstate) $ and thus not extremal in $ \clSeto(\pstate) $.
					
			\item $ \exists l $ satisfying $m<l$: $l$ initially has a smaller slope than $m$ in $ \pstate $ and thus swapped with $ m $ during $ \vec{\beta}^{(\text{NS})} $. Denote the last such $ l $ swapped with $ m $ as $ l_{0} $. 
			\begin{enumerate}
				\item If $ l_{0} $ and $ m $ are neighbouring in $ \vec{\beta}^{(\text{NS})}\pstate $, standard formation indicates
				\begin{equation}\label{key}
					\vec{\beta}^{(\text{NS})} = \vec{\beta}^{(\text{Irrel})}\left(\prod_{n_{i}\in N}\beta^{(n_{i},m)}\right)\left(\prod_{n_{i}\in N}\beta^{(n_{i},l_{0})}\right)\beta^{(m,l_{0})}\vec{\beta}^{(\text{Pre})},
				\end{equation}
				for some set of levels $ N\subset\{n\vert n<m\} $. Here, $ \vec{\beta}^{(\text{Irrel})} $ acts on neither $ m $ nor $ l_{0} $.
				Then as in Eq.~\eqref{eq:rearrange_using_comm}, 
				\begin{align}
					\beta^{(k,m)}\vec{\beta}^{(\text{NS})} &= \beta^{(k,m)}\beta^{(m,l_{0})}\vec{\beta}^{(\text{Irrel})}\left(\prod_{n_{i}\in N}\beta^{(n_{i},l_{0})}\right)\left(\prod_{n_{i}\in N}\beta^{(n_{i},m)}\right)\vec{\beta}^{(\text{Pre})}\nonumber\\
					 &= \beta^{(k,m)}\beta^{(m,l_{0})}\vec{\beta}^{(\text{Pre-}2)}.\label{eq:rearrange_using_comm2}
				\end{align}
				$ \border(\vec{\beta}^{(\text{Pre-}2)}\pstate) = (\cdots,k,\cdots,m,l_{0},\cdots) $ and thus $ \beta^{(k,m)}\beta^{(m,l_{0})} $ produces a non-extreme state from Lemma~\ref{rmk:only_one_nn}.
				\item If $ l_{0} $ and $ m $ are not neighbouring and there is no $ l $ satisfying $k<l<m$ or $k>l>m$, the levels between $ l_{0} $ and $ m $ are $ \{j_{i}\} $ such that $ j_{i}<k,m $. By the construction of the standard formation and the resulting order  
				\begin{equation}\label{key}
					\slope(\vec{\beta}^{(\text{NS})} \pstate)_{k}>\slope(\vec{\beta}^{(\text{NS})} \pstate)_{l_{0}}\geq\slope(\vec{\beta}^{(\text{NS})} \pstate)_{j_{i}}>\slope(\vec{\beta}^{(\text{NS})} \pstate)_{m},
				\end{equation}
			we know that for all $ j_{i} $: i) $ \beta^{(j_{i},l_{0})} $ and $ \beta^{(j_{i},k)} $ exist in $ \vec{\beta}^{(\text{NS})} $, ii) $ \beta^{(j_{i},k)} $ proceeds $ \beta^{(j_{i},l_{0})} $, and iii) $ \beta^{(j_{i},m)} $ does not exist in $ \vec{\beta}^{(\text{NS})}$. The last condition also implies that after $ \beta^{(j_{i},l_{0})} $, swaps acting on $ j_{i} $ are only acting within the set $ \{j_{i}\} $, which we will denote as $ \vec{\beta}^{(\{j_{i}\})} $.
		Then 
			\begin{align}
				\beta^{(k,m)}\vec{\beta}^{(\text{NS})} &= \beta^{(k,m)}\vec{\beta}^{(\text{Irrel})}\vec{\beta}^{(\{j_{i}\})}\left(\prod_{i}\beta^{(j_{i},l_{0})}\right)\beta^{(l_{0},m)}\vec{\beta}^{(\text{Pre})}\nonumber\\
				 &=  \vec{\beta}^{(\{j_{i}\})}\left(\prod_{i}\beta^{(j_{i},l_{0})}\right)\beta^{(k,m)}\vec{\beta}^{(\text{Irrel})}\beta^{(l_{0},m)}\vec{\beta}^{(\text{Pre})},
			\end{align}
			since $ \beta^{(k,m)}\vec{\beta}^{(\text{Irrel})} $ does not act on levels $ l_{0} $ and all $ j_{i} $. Finally, $ l_{0} $ is again neighbouring to $ m $ in $ \border(\vec{\beta}^{(\text{Irrel})}\beta^{(l_{0},m)}\vec{\beta}^{(\text{Pre})}\pstate) $ and we can use the argument from case ii.A to prove this state is non-extremal.
			\end{enumerate}
			\item All $ l<k,m $: denote the last such $ l $ swapped with $ k $ as $ l_{0} $ and they are neighbouring from the structure of the standard formation. Similar to Eq.~\eqref{eq:rearrange_using_comm2}, we get $ \beta^{(k,m)}\beta^{(k,l_{0})} $ part, which produces a non-extreme state by Lemma~\ref{rmk:only_one_nn}.
		\end{enumerate}
		As a result, non-neighbouring swaps are completely ruled out from the candidate of extreme point producing $\beta$-swaps when starting from monotonic order states. 
	\end{enumerate}
\end{enumerate}

	The sufficient condition of the lemma can be shown by recalling two properties: i) there exists at least one extreme point of $\clSeto$ for each $\beta$-order (Lemma~\ref{lemma:ext_point_for_each_order}) and ii) $\beta$-swap series satisfying the conditions of the lemma are all equivalent if the pair ($\pi,\pi^\prime$) is identical (Lemma~\ref{lemma:standard_form}), which makes them the only candidate for an extreme point with order $ \pi^{\prime} $. 
	
	Lastly, we note that the proof holds for $\pi = (d,d-1,\cdots,1)$ initial states since all we have used are Remark~\ref{rmk:swap_series_eq}, Lemma~\ref{rmk:only_one_nn}, and Lemma~\ref{lemma:standard_form}, which hold even when the energy ordering is inverted. 

\end{proof}

\section{Proof of Theorem~\ref{thm:ETO_cone}}\label{appendix:ETO_cone_proof}
Before proving Thm.~\ref{thm:ETO_cone}, we establish a result connecting lower dimensional results to higher dimensions. Lemma~\ref{lemma:high_d_two_swaps}, combined with the fact that two disjoint neighbouring $ \beta $-swaps produce $ \clSto $ extreme points (using Remark~\ref{lemma:neighbouring_TOext} twice), shows that states after two different neighbouring $ \beta $-swaps are always uniquely extremal in their order.

\begin{lemma}
	For any $ \pstate\in\probspace^d $ with the corresponding $ \beta $-ordering $ \pi = \border(\pstate) $, and for any index $ j\in [2,d-1] $, the extreme states of $ \clSeto(\pstate) $ that have $ \beta $-orders
	\begin{align}\label{key}
		S_{j-1,j}S_{j,j+1}\pi = (\pi_1,\cdots,\pi_{j+1},\pi_{j-1},\pi_{j}, \cdots,\pi_d), \\ 
		S_{j,j+1}S_{j-1,j}\pi = (\pi_1,\cdots,\pi_{j},\pi_{j+1},\pi_{j-1},\cdots,\pi_d), 
	\end{align} 

are unique and given by the respective states
\begin{align}
	\qstate^{(j,j-1)} &= \beta^{(\pi_{j-1},\pi_{j+1})}\beta^{(\pi_j,\pi_{j+1})}\pstate, \\
	 \qstate^{(j-1,j)} &= \beta^{(\pi_{j-1},\pi_{j+1})}\beta^{(\pi_{j-1},\pi_{j})}\pstate.
\end{align}
	\label{lemma:high_d_two_swaps}
\end{lemma}
\begin{proof}
	One can prove the above statement for the following independent cases: 
	\begin{enumerate}
		\item 
		$S_{j-1,j}S_{j,j+1}\pi $ and $ S_{j,j+1}S_{j-1,j}\pi $ admit unique extreme points for $ \pi_{j}< \pi_{j+1},\pi_{j-1} $: out of the three consecutive levels in $ \beta $-ordering, the middle term corresponds to the lowest energy level (e.g. a $ \beta $-ordering (2,1,3)),
		\item  $ S_{j-1,j}S_{j,j+1}\pi $ for $ \pi_{j-1}< \pi_{j},\pi_{j+1} $,
		\item $ S_{j,j+1}S_{j-1,j}\pi $ for $ \pi_{j+1}<\pi_{j-1},\pi_{j} $, 
		\item $ S_{j,j+1}S_{j-1,j}\pi $ for $ \pi_{j-1}< \pi_{j},\pi_{j+1} $, 
		\item $ S_{j-1,j}S_{j,j+1}\pi $ for $ \pi_{j+1}<\pi_{j-1},\pi_{j} $.
\end{enumerate}
	Cases 1-3 follow from a slight generalization of Remark~\ref{rmk:two_swap_TO_ext}. If a matrix $M$ is a biplanar extremal thermal process (see~\cite{Mazurek19_channels} Def. 6 and Sec. IV), then $M\oplus \mathbb{1} $ is also a biplanar extremal thermal process, since an identity connects $ i $'th element only to $ i $'th element and thus biplanar; and all elements of thermal processes are non-negative and upper bounded by $ 1 $, making an identity extremal.
	Hence, we only need to prove for cases 4 \& 5.
	
	Case 4: 
This can be proven by showing that all states $ \qstate'\in\clSeto(\pstate) $ such that $ \border(\qstate') = \border(\qstate^{(j-1,j)}) $ are thermomajorised by $ \qstate^{(j-1,j)} $. 
	We will show this by contradiction. To do so, first note that $ \qstate^{(j-1,j)} $ is very close to being tightly-majorised by the initial $ \pstate $ -- in particular, all elbow points of $ \mathcal{L}_{\qstate^{(j-1,j)}} $ lie on $ \mathcal{L}_{\pstate} $ except for the $ j $'th one. 
	Therefore, for any $ \qstate\in\clSeto(\pstate) $ with the same $ \beta $-ordering as $  \qstate^{(j-1,j)} $, their elbows are aligned, and we know that $ \mathcal{L}_{\qstate^{(j-1,j)}} \geq \mathcal{L}_{\qstate'} $ at all elbows other than $ j $'th. 
	Now, assume that in fact $ \qstate^{(j-1,j)}\nsucc_\beta\qstate' $. However, this can only happen if $ \mathcal{L}_{\qstate'} > \mathcal{L}_{\qstate^{(j-1,j)}} $ at the $ j $'th elbow.
	Since $ \border(\qstate^{(j-1,j)}) = (\cdots,\pi_{j},\pi_{j+1},\pi_{j-1},\cdots) $, the above condition translates into 
		\begin{equation}
			q^{(j-1,j)}_{\pi_{j-1}} + \sum_{k=j+2}^{d}q^{(j-1,j)}_{\pi_k} =  q^{(j-1,j)}_{\pi_{j-1}} + \sum_{k=j+2}^{d}p_{\pi_k} > q'_{\pi_{j-1}} + \sum_{k=j+2}^{d}q'_{\pi_k}, \label{eq:reducing_pops}
		\end{equation}
	where we have simply used the fact that $  \mathcal{L}_{\qstate'} $ at the $ j $-th elbow is equal to one minus the probability mass corresponding to the $ j+1 $-th up to $ d $-th elements in the $ \beta $-order $ \border(\qstate') $.

Let us denote  $A_R= \sum_{k\in R}p_{k} $ to be the probability mass over a set of levels $ R $. Note that this quantity changes whenever a level $ i\in R $ is $ \beta $-swapped with a level in $ i'\in R^{\mathsf{c}} $. When level $ i' $ has a steeper slope than $ i $, then $ A_R$ increases after the $ \beta $-swap; conversely, $ \sum_{k\in R}p_{k} $ decreases. To make use of this observation for Eq.~\eqref{eq:reducing_pops}, take
	\begin{equation}\label{key}
	 R = \{\pi_{l}\vert l=j-1 \text{ or } l\geq j+2\} .
	\end{equation} 
The only $ \beta $-swaps that can reduce $ A_{R} $ would be $ \beta^{(\pi_{j-1},\pi_{j})} $ and  $ \beta^{(\pi_{j-1},\pi_{j+1})} $.
Among different combinations of two swaps, $ \qstate^{(j-1,j)} = \beta^{(\pi_{j-1},\pi_{j+1})}\beta^{(\pi_{j-1},\pi_{j})}\pstate $ achieves the minimum probability mass over a set of levels $ R $ as can be seen from Thm.~\ref{thm:3dETO}. Thus, Eq.~\eqref{eq:reducing_pops} is impossible. 


	Case 5: $ \pi_{j+1}< \pi_{j-1},\pi_{j} $, we must try to maximise $ \sum_{k=1}^{j-2}q_{\pi_k} + q_{\pi_{j+1}} $, which is achieved when $ \qstate = \qstate^{(j,j-1)} $ as in Case 4. This concludes the proof.
\end{proof}

Now we prove Thm.~\ref{thm:ETO_cone}.

\begin{proof}
Given an initial state $ \pstate_0 \in\mathcal{V}^d$, suppose that a $\beta$-swap series $\vec{\beta} = \beta^{(j_l,k_l)}\cdots\beta^{(j_1,k_1)}$ of length $ l $ produces an extreme point $\pstate_l = \vec{\beta}\pstate_0$ of $\clSeto(\pstate_0)$. To prove Thm. \ref{thm:ETO_cone} is to prove that $ \ell \leq  \left\lfloor\frac{d!-4}{d-3}\right\rfloor$. 

We start by denoting 
\begin{equation}\label{eq:series_of_ext}
	\pstate_n \equiv \beta^{(j_n,k_n)}\cdots\beta^{(j_1,k_1)} \pstate_0 \qquad \text{for } n=1,\cdots, \ell. 
\end{equation}
Because of the assumption that $ \pstate_l $ is extremal, by necessity all the intermediate states $ \pstate_1,\cdots,\pstate_{\ell-1} $ are also extreme points of $ \clSeto(\pstate_0) $. Furthermore, we may use the fact that 
\begin{equation}\label{eq:no_same_beta_order}
	\border(\pstate_n)\neq\border(\pstate_m)\qquad \text{ for all }n\neq m,
\end{equation}
since together with Eq.~\eqref{eq:series_of_ext} this would imply that $ \pstate_{n}\succ_{\beta}\pstate_{m} $ if $ n<m $. Then the original bound $\ell_{\rm{max}} \leq d! - 1$ of~\cite{Lostaglio_18_ETO} is immediately obtained. 

This bound can be improved by carefully studying the total number of feasible $\beta$-orders that $\pstate_{n+1}$ can take, given its history $(\pstate_0,\cdots,\pstate_n)$, which we denote as $R(n)$. For instance, $R(0)\leq d! - 1$, where the upper bound refers to choosing any $\beta$-order for $\pstate_1$ except $\border(\pstate_0)$. This is an upper bound, likely to be a loose one, because for a fixed $ \pstate $, not all $ R(0) $ orders are possible, given that $ \pstate_1 $ is obtained only by a single $ \beta $-swap.
If $R(n)<1$ for all trajectories $(\pstate_0,\cdots,\pstate_n)$, then $p_{n+1}$ cannot be extremal and therefore $\ell_{\rm max}\leq n$.

Now, given $\pstate_1$, we know the following:
\begin{enumerate}
	\item The state $ \pstate_2 \in \clSeto(\pstate_1) \subsetneq \clSeto(\pstate_0)$,
	\item Since $ \pstate_2 $ is extremal for $\clSeto(\pstate_0)$, it is also extremal for $\clSeto(\pstate_1)$,
	\item By Lemma~\ref{lemma:neighbouring_TOext}, the states $\qstate_0^{(j)} = \beta^{(\border(\pstate_0)_j,\border(\pstate_0)_{j+1})}\pstate_0$ are unique extreme points of $\clSeto(\pstate_0)$ with order $\border(\qstate_0^{(j)})$.
\end{enumerate}  
Firstly, by Eq.~\eqref{eq:no_same_beta_order}, $\border(\pstate_2) \neq \border(\pstate_0),\border(\pstate_1)$, leading to $R(1)\leq d! - 2$. The third observation above allows us to conclude that $\border(\pstate_2) \neq \border(\qstate_0^{(j)})$ for $j = 1,\cdots, d-1$. However, we should not subtract $d-1$ orders from $R(1)$ for all $(\pstate_0,\pstate_1)$. When $\pstate_1 = \qstate_0^{(k)}$ for some $k$, $\border(\qstate_0^{(k)})$ and $\border(\pstate_1)$ must not be double-counted, giving upper bound $R(1)\leq d!-d$ for any trajectory $(\pstate_0,\pstate_1)$.

Similar as to the above reasoning, for $ \pstate_{i+1} $, we know that 
	\begin{enumerate}
		\item The state $ \pstate_{i+1} \in \clSeto(\pstate_{i}) \subsetneq \cdots \subsetneq \clSeto(\pstate_0)$,
		\item Since $ \pstate_{i+1}\in\mathrm{Ext}(\clSeto(\pstate_0)) $, therefore $ \pstate_{i+1}\in\mathrm{Ext}(\clSeto(\pstate_{k})) $ for all $ k\leq i $,
		\item By Lemma~\ref{lemma:neighbouring_TOext}, $ \qstate_{i}^{(j)} = \beta^{(\border(\pstate_{i})_{j}, \border(\pstate_{i})_{j+1})}\pstate_{i} $ are unique extreme points of $ \clSeto(\pstate_{i}) $ with order $ \border(\qstate_{i}^{(j)}) $.
	\end{enumerate} 
Therefore, given $(\pstate_0,\cdots,\pstate_n)$, we can repeat the same using $\border(\pstate_{n+1}) \neq \border(\pstate_{k})$ for $k\leq n$, and $\border(\pstate_{n+1}) \neq \border(\qstate_{k}^{(j)})$ for $k\leq n-1$ and $1\leq j\leq d-1$. The former condition excludes $n+1$ orders and the latter one occupies $n(d-1)$ orders. Again, we need to account for double counting. 
\begin{enumerate}
	\item $\border(\pstate_i) = \border(\qstate_{i^\prime}^{(j)})$ for some $i,i^\prime$, and $j$.
	\item $\border(\qstate_i^{(j)}) = \border(\qstate_{i^\prime}^{(j^\prime)})$ for some $j, j^\prime$ and $i^\prime>i$.
\end{enumerate}
We will show that above coincidences can happen only when $\pstate_{i^\prime\pm1} = \qstate_{i^\prime}^{(j)}$ (case 1), or $\qstate_{i^\prime-2}^{(j,j^\prime)} = \pstate_{i^\prime}$ (case 2).

If $\pstate_{i} = \qstate_{i^\prime}^{(j)}$ for $i\leq i^\prime -2$, $\border(\pstate_{i^\prime}) = S_{j,j+1}(\border(\pstate_i)) = \border(\qstate_i^{(j)})$. But $\border(\qstate_i^{(j)})$ is already removed from $ R(i^\prime-1) $ for all $i^\prime\geq i+2$. $\pstate_{i} = \qstate_{i^\prime}^{(j)}$ for $i\geq i^\prime +2$ is also not allowed from the same reason.

The second case implies that $S_{j,j+1}(\border(\pstate_i)) = S_{j^\prime,j^\prime+1}(\border(\pstate_i^\prime))$ or equivalently, there exists some $ \beta $-ordering $ \chi $ such that 
\begin{equation}\label{key}
\chi =	S_{j^\prime,j^\prime+1}(S_{j,j+1}(\border(\pstate_i)))= \border(\pstate_i^\prime).
\end{equation}
Consider then three cases: 
\begin{enumerate}[label=\alph*)]
	\item $j^\prime = j$: $\border(\pstate_i) = \border(\pstate_{i^\prime})$ and $\pstate_{i^\prime}$ is not extremal for $\clSeto(\pstate_0)$.
	\item $j^\prime = j\pm1$: by 
	Lemma~\ref{lemma:high_d_two_swaps}, we know that $\clSeto(\pstate_i)$ has a unique extreme point that has the $ \beta $-ordering $ \chi $. This is given by the state $\qstate_{i}^{(j,j^{\prime})}$, that is, $\pstate_{i^\prime} = \qstate_{i}^{(j,j^{\prime})}$, $i^\prime = i+2$, and $ \border(\qstate_i^{(j)}) = \border(\qstate_{i^\prime}^{(j^\prime)}) = \border(\pstate_{i+1}) $. We do not have to worry about double counting $ \border(\qstate_{i^\prime}^{(j^\prime)}) $ since it is already dealt as case 1: $\border(\pstate_{i^\prime-1}) = \border(\qstate_{i^\prime}^{(j^\prime)})$.
	
	\item if $j^\prime \neq j, j\pm1 $, $\beta^{(j^\prime,j^\prime+1)}$ and $\beta^{(j,j+1)}$ commute and $\beta^{(j^\prime,j^\prime+1)}\beta^{(j,j+1)}\pstate_i$ is tightly thermomajorised by $\pstate_i$, again, making $\pstate_{i^\prime} = \beta^{(j^\prime,j^\prime+1)}\beta^{(j,j+1)}\pstate_i$ and $i^\prime = i+2$. Since $ \beta^{(j,j+1)} $ and $ \beta^{(j^{\prime},j^{\prime}+1)} $ commute, $ \pstate_{i+1} $ can be both $ \beta^{(j,j+1)}\pstate_i $ or $  \beta^{(j^\prime,j^\prime+1)}\pstate_i $. 
	We do not have to worry about the first case since $ \border(\qstate_i^{(j)}) = \border(\qstate_{i^\prime}^{(j^\prime)}) = \border(\pstate_{i+1}) $ as before. In the other case, $ \border(\qstate_i^{(j)}) = \border(\qstate_{i^\prime}^{(j^\prime)}) \neq \border(\pstate_{i+1}) $ and we should not count $ \border(\qstate_{i^\prime}^{(j^\prime)})  $ again.
\end{enumerate}

Then, newly removed orders in $ R(i^\prime+1) $ are: $ \border(\pstate_{i^\prime+1}) $, $ \border(\qstate_{i^\prime}^{(k)})$. The worst case is when $ \pstate_{i^{\prime}} = \qstate_{i^{\prime}-2}^{(j^{\prime},j)} $, $ \pstate_{i^{\prime}-1} = \qstate_{i^{\prime}-2}^{(j^{\prime})} $, and $ \pstate_{i^{\prime}+1} = \qstate_{i^{\prime}}^{(j^{\prime\prime})} $ with $ j,j+1,j^{\prime},j^{\prime}+1,j^{\prime\prime} $ all distinct, which then gives  \begin{equation}\label{key}
		\border(\qstate_{i^{\prime}}^{(j)}) = \border(\pstate_{i^{\prime}-1}) ,\quad  \border(\qstate_{i^{\prime}}^{(j^{\prime})}) = \border(\qstate_{i^{\prime}-2}^{(j)}) ,\quad  \border(\qstate_{i^{\prime}}^{(j^{\prime\prime})}) = \border(\pstate_{i^{\prime}+1}).
	\end{equation}
Excluding these double-counted orders, one can remove at least $ d-3 $ orders at each step.

To sum up, $R(2) \leq d! - d - (d-2)$ since $\qstate_1^{(j)}$ eliminates $d-2$ new levels. For $n>3$, at least $d-3$ new $\beta$-orders can be eliminated from $R(n)$ after applying each $\beta$-swap.  As a result, we get
\begin{equation}
	R(n) \leq d! - 4 - (d-3)n,\quad n\geq 3\ \text{and}\ d\geq 4,
\end{equation}
which bounds the length of $\beta$-swap series as $l_{\rm max} = \left\lfloor\frac{d!-4}{d-3}\right\rfloor$ when $ d\geq 4 $.

\end{proof}

\section{Catalytic elementary thermal operations: an example}\label{appendix:CETO}

We tackle the problem by the following procedure: given $ \pstate\in \mathcal{V}^3$, we choose a qubit catalyst $ \catstate \in \mathcal{V}^2 $, and construct the reachable state set $\clSeto(\pstate\otimes\catstate)$. The catalyst Hamiltonian is assumed to be w.l.o.g. degenerate, and hence a choice of $ \catstate = (c_1,1-c_1) $ is simply determined by a real-valued parameter $ c_{1} $. We then find the full characterisation of $ \clSeto(\pstate\otimes\catstate) $ by identifying its extreme points, and denote $\clScetoqb(\pstate;\catstate)$ to be the set of all states $\qstate$ such that  $\qstate\otimes\catstate\in\clSeto(\pstate\otimes\catstate)$. 

Lastly, this procedure is iterated for different choices of $\catstate$, by varying the choice of $ c_1 $ in a sufficiently fine-grained manner. 

The first step, where $\clSeto(\pstate\otimes\catstate)$ is found, is done with the first algorithm of Sec.~\ref{subsec:numerics_high_d} where convex hull is explicitly constructed. Since degenerate Hamiltonian is assumed for catalysts, this process is easier than a generic search. The resulting set $\clSeto(\pstate\otimes\catstate)$ also encompasses states $p'\in\mathcal{V}^6$  that are not in product form. Final states of the form $ \qstate\otimes\catstate $ can be distilled by forming plane equations imposing catalyst state $\catstate$, and performing half-space intersections with the full set $\clSeto(\pstate\otimes\catstate)$ numerically. In general, extreme points of $\clScetoqb(\pstate;\catstate)$ are not necessarily extremal in $\clSeto(\pstate\otimes\catstate)$.  

Note that the catalytic ETO operation that produces $ \pstate\rightarrow\qstate $ is often non-unique. First of all, there may be multiple catalyst states that enable a particular transition. Nevertheless, given $ \pstate,\qstate $ and a choice of $ \catstate $, if $ \qstate\otimes\catstate \in \clSeto(\pstate\otimes\catstate) $, then decomposing the transformation into a particular convex combination of different $\beta$-swap series is straightforward: writing $ \qstate\otimes\catstate $ as a convex combination of $\mathrm{Ext}[\clSeto(\pstate\otimes\catstate)] $ can be done by linear programming and each extreme point is obtained from a $\beta$-swap series.
This gives us a particular algorithm for implementation.

Now we will examine one specific extreme point $ \clScetoqb(\pstate;\catstate) $, plotted in Figs.~\ref{fig:CETO_freeE} and~\ref{fig:local_freeEs}, following the procedure described above. 
This example is also an extreme point of $ \clScetoqb(\pstate) $ having the smallest ground state population among reachable states.

\begin{tcolorbox}[breakable, colback=red!5!white,colframe=red!65!black,fonttitle=\bfseries,boxrule=0.7pt, title = An extreme state of $ \clScetoqb $]

The initial state $ \pstate = (0.35,0.55,0.1) $ and the Hamiltonian $\beta\mathrm{H}_S = (0,0.2,0.5)$, giving $\border(\pstate) = (2,1,3)$. The initial catalyst distribution is fine-tuned to be $\catstate = (c_1,1-c_1)$, with
\begin{equation}
	c_{1} = \frac{-p_3+\sqrt{p_3^2+8\Delta_{13}p_1p_3}}{4\Delta_{13}p_1} \simeq 0.3816,
\end{equation}
which provides the maximum advantage when minimizing the ground state population. Then the total $\beta$-order becomes
\begin{equation}
	\border(\pstate\otimes\catstate) = (2*2,2*1,1*2,1*1,3*2,3*1),
\end{equation}
where $a*b$ implies energy eigenstate $\lvert a\rangle_{S}\lvert b\rangle_{C}$ of a system (S) plus catalyst (C) state.

From $ \clSeto(\pstate\otimes\catstate) $, we obtain new extreme points of $ \clScetoqb(\pstate;\catstate) $, including our example point $ \qstate $. Further, by solving a linear programming problem, four extreme states of $ \clSeto(\pstate\otimes\catstate) $, $\qstate^\prime_{1,2,3,4}\in\probspace^6$, such that 
\begin{equation}
	\qstate\otimes\catstate = \sum_{i=1}^{4}\alpha_i\qstate^\prime_i,\quad \sum_{i=1}^{4}\alpha_i = 1,\quad \alpha_i\geq 0,
\end{equation}
and corresponding $\beta$-swaps $ \qstate^\prime_i = \vec{\beta}_i(\pstate\otimes\catstate) $ can be identified. $ \vec{\beta}_i $ differ only slightly from each other (differences marked red),
\begin{align}
	\vec{\beta}_1 &= \beta^{(1*2,3*1)} \beta^{(1*1,3*1)}\beta^{(2*1,3*2)}\beta^{(1*2,3*2)}\beta^{(1*1,3*2)},\\
	\vec{\beta}_2 &= \beta^{(1*2,3*1)} \beta^{(1*1,3*1)}{\color{red!65!black}\beta^{(2*1,2*2)}}\beta^{(2*1,3*2)}\beta^{(1*2,3*2)}\beta^{(1*1,3*2)},\\
	\vec{\beta}_3 &= \beta^{(1*2,3*1)} \beta^{(1*1,3*1)}{\color{red!65!black}\beta^{(2*2,3*2)}}\beta^{(2*1,3*2)}\beta^{(1*2,3*2)}\beta^{(1*1,3*2)},\\
	\vec{\beta}_4 &= \beta^{(1*2,3*1)} \beta^{(1*1,3*1)}{\color{red!65!black}\beta^{(2*1,2*2)}\beta^{(2*2,3*2)}}\beta^{(2*1,3*2)}\beta^{(1*2,3*2)}\beta^{(1*1,3*2)},
\end{align}
facilitating the recombination into
\begin{equation}
	\qstate\otimes\catstate = \beta^{(1*2,3*1)} \beta^{(1*1,3*1)}\Mswap_{\lambda_2}^{(2*1,2*2)}\Mswap_{\lambda_1}^{(2*2,3*2)}\beta^{(2*1,3*2)}\beta^{(1*2,3*2)}\beta^{(1*1,3*2)}(\pstate\otimes\catstate),\label{eq:example_CETO_path}
\end{equation}
with $ \lambda_1 = \alpha_3/(\alpha_1+\alpha_3) = \alpha_4/(\alpha_2+\alpha_4) $ and $ \lambda_2 = \alpha_2/(\alpha_1+\alpha_2) = \alpha_4/(\alpha_3+\alpha_4) $. 

Given Eq.~\eqref{eq:example_CETO_path}, we can analyse system-catalyst interplay during the catalytic evolution, starting from
\begin{equation}
	\pstate\otimes\catstate \simeq (0.1336, 0.2164, 0.2099, 0.3401, 0.0382, 0.0618).
\end{equation}
The swap series can be grouped into three phases:
\begin{enumerate}
	\item The first four swaps all involve the $3*2$ population. The first and the third swaps ($ \beta^{(1*1,3*2)} $ and $ \beta^{(2*1,3*2)} $) work to shift population from the first level of the catalyst to the second level and thus intensify the non-uniformity of the catalyst reduced state, as reflected in corresponding steps of Fig.~\ref{fig:local_freeEs} (b). 
	The ratio between $3*2$ and $3*1$ populations is increasing more rapidly than the one between levels $ 2 $ and $ 1 $ of catalyst reduced state, correlating the system and catalyst as shown in the mutual information.
	The fourth swap is chosen to be $ \lambda\neq1 $ swap to prevent $3*2$ population to become too large to recover $ \catstate $.

	Since both catalyst local free energy and mutual information increase, system local free energy should always decrease at this stage. The total state after these swaps is
	\begin{align}
		\rstate^\prime_1 &\simeq (0.1144, 0.1662, 0.1857, 0.3323 , 0.0382, 0.1633),\\ 
		\Tr_{C}[\rstate^\prime_1] &\simeq (0.2806, 0.5180, 0.2015),\quad \Tr_{S}[\rstate^\prime_1] \simeq (0.3382, 0.6618)
	\end{align}
	\item The fifth swap balances catalyst distribution in the degenerate block $\ket{2}_S$ while system reduced populations are fixed, yielding
	\begin{align}
		\rstate^\prime_2 &\simeq (0.1144, 0.1662, 0.1977, 0.3203 , 0.0382, 0.1633),\\ 
		\Tr_{C}[\rstate^\prime_2] &\simeq (0.2806, 0.5180, 0.2015),\quad \Tr_{S}[\rstate^\prime_2] \simeq (0.3502, 0.6498).
	\end{align}
	Catalyst local free energy decreases as a result of mixing, while correlation increases. 
	\item Last two swaps increases $3*1$ level population to recover the original ratio between $3*1$ and $3*2$, while at the same time further reducing the system ground state population. For this particular choice of catalyst, we have a simplification in the sense that these swaps also balance $1*1$ and $1*2$, leading to 
	\begin{align}
		\qstate\otimes\catstate &\simeq (0.0832, 0.1348, 0.1977, 0.3203, 0.1008,	0.1633), \\
		\qstate &\simeq (0.2179, 0.5180, 0.2641),
	\end{align}
	with vanishing correlation and retrieval of the original catalyst. The system free energy increases here, since level $\ket{1}_S$, which already has the lowest slope, loses more population and the new state will thermomajorises the old state.
	This behaviour of increasing free energy after swaps is strictly forbidden in non-catalytic setting, and allowed in this case by sacrificing correlations and catalyst free energy stored from previous operations. 
\end{enumerate}
\end{tcolorbox}

Although there is a room for some change in orders between swaps that are commuting, the above is a typical strategy for constructing catalytic transformations: 
\begin{enumerate}
	\item exploits expanded dimensionality to swap a system plus catalyst level with more numbers of levels, paying i) temporary correlations and ii) local variations on catalyst as a cost; 
	\item resolves correlations by mixing in the degenerate energy subspace; and
	\item recovers original catalyst distribution while increasing system local free energy.
\end{enumerate}

\begin{figure}[h!]\centering
	\includegraphics[width=\columnwidth]{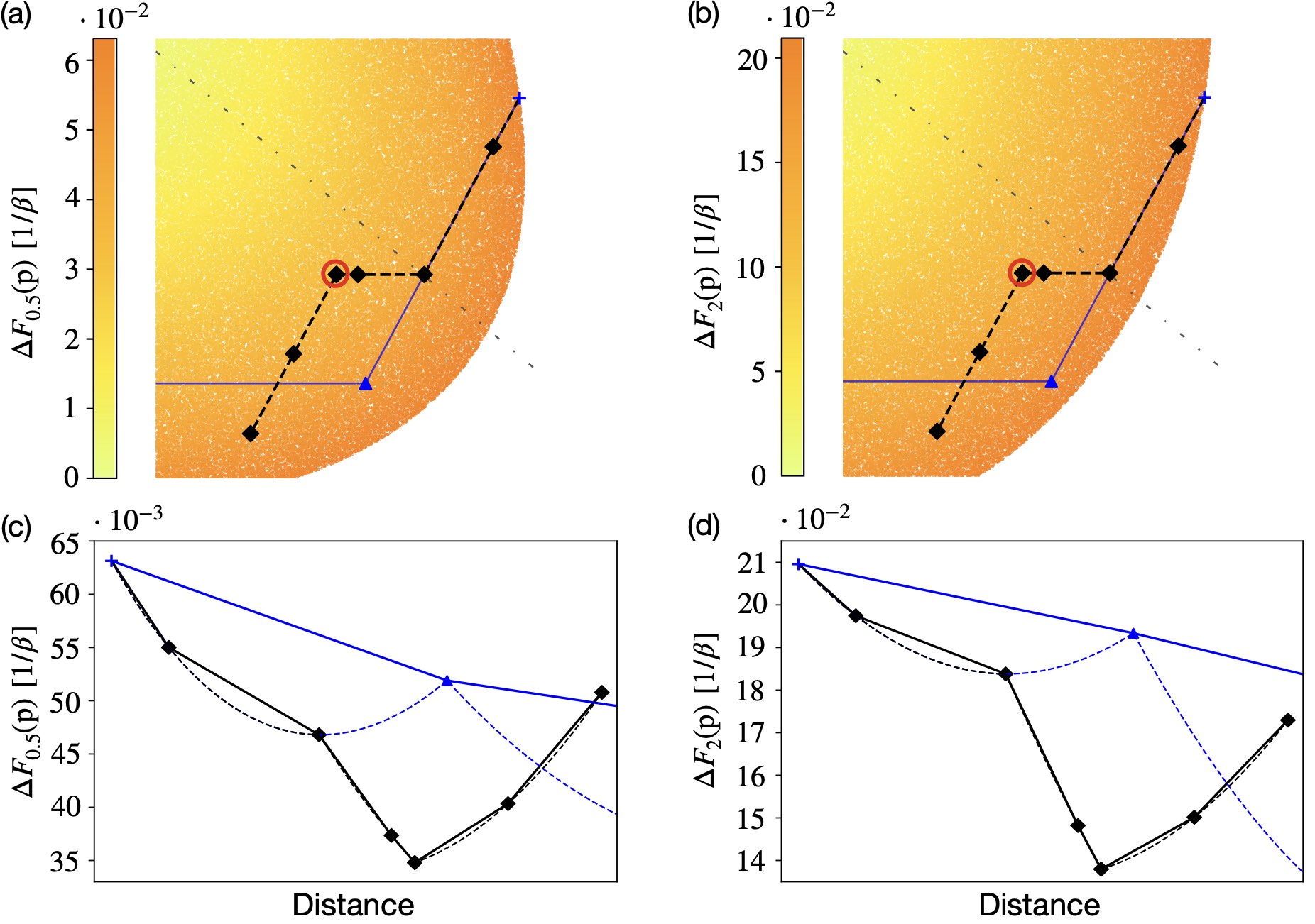}
	\caption{Replica of Fig.~\ref{fig:CETO_freeE} with generalised free energies.}
	\label{fig:generalised_free_Es}
\end{figure}

Even for generalised free energies $F_\alpha$ with different $\alpha$ values, similar behaviours are observed. In Fig.~\ref{fig:generalised_free_Es}, $\alpha = 0.5$ and $2$ are presented as representative examples. In both of the cases, $F_\alpha$ for a system reduced state (see (c) and (d) of Fig.~\ref{fig:generalised_free_Es}) shares the decreasing/increasing trend, albeit with different slopes.

\end{document}